\documentclass[12pt,a4paper]{article}

\usepackage[ruled]{algorithm2e} 

\SetAlFnt{\small}
\SetAlCapFnt{\small}
\SetAlCapNameFnt{\small}

\SetAlCapHSkip{0pt}
\IncMargin{-\parindent}

\usepackage{times}
\usepackage{graphicx}
\usepackage{amssymb}
\usepackage{dsfont}
\usepackage{complexity}
\usepackage{bbold}

\usepackage{amsmath}
\usepackage{amsthm}
\usepackage{amsfonts}
\usepackage{enumerate}
\usepackage{hyperref}
\usepackage[margin=1in]{geometry}

\usepackage{tikz}

\newtheorem{theorem}{Theorem}
\newtheorem*{theorem*}{Theorem}
\newtheorem{definition}{Definition}
\newtheorem{lemma}[theorem]{Lemma}

\newtheorem{corollary}[theorem]{Corollary}

\newtheorem{example}{Example}

\newtheorem{prop}[theorem]{Proposition}

\newcommand{\weaklypref}{\succeq}

\newcommand{\pref}{\succ}
\newcommand{\tie}{\simeq}
\newcommand{\agentset}{N}

\DeclareMathOperator{\worst}{worst}
\DeclareMathOperator{\secondworst}{second\_worst}

\DeclareMathOperator{\type}{type}
\DeclareMathOperator{\subtype}{subtype}
\DeclareMathOperator{\assigned}{assigned}

\DeclareMathOperator{\exworst}{worst^{ex}}
\DeclareMathOperator{\cls}{class}
\DeclareMathOperator{\ext}{ex}

\newcommand{\cplxty}[1]{\textup{\textsf{#1}}}

\title{Solving Hard Stable Matching Problems \\ Involving Groups of Similar Agents}
\author{
   Kitty Meeks\\
   School of Computing Science\\
   University of Glasgow\\
   Glasgow, UK\\
   \emph{kitty.meeks@glasgow.ac.uk}
   \and Baharak Rastegari\\
   Department of Electronics and Computer Science\\
   University of Southampton\\
   Southampton, UK\\
   \emph{b.rastegari@soton.ac.uk}
}

\begin{document}		
\date{}
\maketitle

\begin{abstract}
Many important stable matching problems are known to be \cplxty{NP}-hard, even when strong restrictions are placed on the input. In this paper we seek to identify structural properties of instances of stable matching problems which will allow us to design efficient algorithms using elementary techniques. We focus on the setting in which all agents involved in some matching problem can be partitioned into $k$ different \emph{types}, where the type of an agent determines his or her preferences, and agents have preferences over types (which may be refined by more detailed preferences within a single type). 
This situation would arise in practice if agents form preferences solely based on some small collection of agents' attributes. We also consider a generalisation in which each agent may consider some small collection of other agents to be exceptional, and rank these in a way that is not consistent with their types; this could happen in practice if agents have prior contact with a small number of candidates. We show that (for the case without exceptions), several well-studied \cplxty{NP}-hard stable matching problems including \textsc{Max SMTI} (that of finding the maximum cardinality stable matching in an instance of stable marriage with ties and incomplete lists) 
belong to the parameterised complexity class \cplxty{FPT} when parameterised by the number of different types of agents needed to describe the instance. 
For \textsc{Max SMTI} this tractability result can be extended to the setting in which each agent promotes at most one ``exceptional'' candidate to the top of his/her list (when preferences within types are not refined), but the problem remains \cplxty{NP}-hard if preference lists can contain two or more exceptions and the exceptional candidates can be placed anywhere in the preference lists, even if the number of types is bounded by a constant.
\end{abstract}

\section{Introduction}\label{sec:introduction}
Matching problems occur in various applications and scenarios such as the assignment of children to schools, college students to dorm rooms, junior doctors to hospitals, and so on. In all the aforementioned, and similar, problems, it is understood that the participants (which we will refer to as agents) have preferences over other agents, or subsets of agents. The majority of the literature assumes that these preferences are ordinal, and that is the assumption we make in this work as well. Moreover, it is widely accepted that a ``good'' and ``reasonable'' solution to a matching problem must be \emph{stable}, where stability is defined according to the context of the problem at hand. Intuitively speaking, a stable solution guarantees that no subset of agents find it in their best interest to leave the prescribed solution and seek an assignment amongst themselves.  Unfortunately, many interesting and important stable matching problems are known to be \cplxty{NP}-hard even for highly restricted cases.

Most hardness results in the study of stable matching problems are based on the premise that agents may have arbitrary preference lists.  In practice, however, agents' preferences are likely to be more structured and correlated. In this work, we consider a setting where agents can be grouped into $k$ different ``types'', where the type of an agent determines (most of) the agent's preferences, and also how s/he is compared against other agents. If we allow each agent to have a different type, this setup does not place any restrictions on the instance. However, we are interested in the setting where the number of types required to describe an instance is much smaller than the total number of agents: such a situation would arise in practice if agents derive their preferences by considering some small collection of attributes of other agents (where each of these attributes has a small number of possible values). As an example, consider the hospitals-residents job market in which junior doctors or residents are to be assigned to hospital posts. It is highly plausible that agents in this market base their preferences on a small collection of candidates' attributes. E.g. hospitals might rank applicants based on their exam grade, interview score, etc, and junior doctors might rank the hospitals based on the programs they offer, their reputation, their geographic location, etc. Similar observations have been made in the literature (see \cite{BGR08,CGM12}) regarding stable marriage market and stable roommates market respectively, where agents form preferences based on candidates' attributes such as attractiveness, intelligence, wealth, etc. In this setting, we obtain our set of types by first partitioning agents by their profile of attributes, then further partitioning each set by the preference list over other profiles of attributes.  Note that the number of possible preference lists depends only on the number of possible attribute profiles.

The notion of types is also useful if we are interested in a relaxation of stability, where agents are only willing to form a private arrangement with a partner who is distinctly superior to their current partner with respect to an important characteristic. It is reasonable to assume that in practice a certain amount of effort is required by both agents in a blocking pair to make a private arrangement outside the matching, and so agents are unlikely to make this effort for a very small improvement in their utility. Suppose that an agent is only willing to make the effort to form a private arrangement if it results in a significantly better partner, specifically one which has a significantly better value for the most important attribute. In this case we only need to consider attributes which are the most important for at least one agent, and moreover we might reasonably consider only a small number of categories of values for these attributes.

The simplest model (discussed in Section \ref{sec:typed} is to assume that the agents of the same type are completely indistinguishable. That is, they have the same preference lists, and every other agent that finds their type acceptable is indifferent between them. Equivalently, we can say that each type has a preference ordering over types of the candidates, which need not be complete or strict. We also consider two generalisations of this basic model.  In the first generalisation (discussed in Section \ref{sec:refined-typed}), agents no longer have to be indifferent between agents of the same type: they can refine their preference lists arbitrarily (so that agents of the same type still occur consecutively), so long as the preference lists for agents of the same type are identical.  In the second generalisation (discussed in Section \ref{sec:exceptions}, we instead enrich the basic model by allowing each agent to consider some small number of other agents ``exceptional'': such agents can appear anywhere in the preference list, regardless of their type.  This situation with exceptions might arise in practice if, for example, an agent knows some of the candidates directly or through a third-party connection and, based on this additional information, ranks them disregarding their type, e.g. at the top or bottom of his/her preference list.

We show that we can solve some of the most important hard stable matching problems efficiently from the point of view of parameterised complexity when the number of types is taken as the parameter. We obtain our results by reducing a given hard stable matching problem to the problem of solving
a number of instances of a tractable
problem, where this number is a function of the number of types.  
Some of our results rely on the fixed parameter tractability of Integer Linear Programming (ILP). We also demonstrate that, by imposing further restrictions, some of these hard stable matching problems become polynomial-time solvable.

\section{Preliminaries}\label{sec:preliminaries}

In this section we introduce the main concepts we use in the paper; we begin with some definitions, then provide a brief introduction to parameterised complexity, before describing existing results on Integer Programming.

\subsection{Definitions}\label{sec:definitions}

In this section we provide the key definitions for the stable matching settings we study; for further background and terminology we refer the reader to \cite{Man13}.

Perhaps the most widely studied matching problem is the \textbf{\emph{Stable Marriage problem (SM)}}. In an instance of SM we have two disjoint sets of agents, men and women, each having a strict preference ordering over the individuals of the opposite sex (candidates). \emph{Stable Marriage with Incomplete lists (SMI)}, \textbf{\emph{Stable Marriage with Ties (SMT)}} and \textbf{\emph{Stable Marriage with Ties and Incomplete lists (SMTI)}} are generalisations of SM where agents are permitted to declare some candidates unacceptable, are allowed to express indifference between two or more candidates, or both, respectively. The \textbf{\emph{Stable Roommates problem (SR)}} is a non-bipartite generalization of SM. Extensions allowing for incomplete lists and indifference in preference lists are defined the same way as for SM. \textbf{\emph{Hospitals$/$Residents problem (HR)}} is a famous extension of SMI that models many practical applications, including the assignment of junior doctors to hospitals, by allowing agents on one side of the market to be assigned to multiple agents on the other side of the market. \textbf{\emph{Hospitals$/$Residents problem with Ties (HRT)}} is an extension of HR that allows for indifference in preference lists (note that in the standard terminology, both HR and HRT allow for incomplete lists, see \cite{Man13}).

Let $\agentset$ denote a set of $n$ agents, which in a bipartite matching setting (i.e. SMTI or HRT) is composed of two disjoint sets. Each hospital $h$ in an instance of HRT is associated with a capacity $q(h)$ that denotes the number of posts it offers.  When in a bipartite matching setting, we use the term \emph{candidates} to refer to the agents on the opposite side of the market to that of an agent under consideration. In non-bipartite settings, candidates refer to all the other agents except the one under consideration.

Each agent finds a subset of candidates acceptable and ranks them in order of preference. Preference orderings need not to be strict, so it is possible for an agent to be indifferent between two or more candidates. We write $b \pref_{a} c$, or equivalently $c \prec_{a} b$, 
to denote that agent $a$ prefers candidate $b$ to candidate $c$, and $b \tie_{a} c$ to denote that $a$ is indifferent between $b$ and $c$. We write $b \weaklypref_{a} c$ to denote that $a$ either prefers $b$ to $c$ or is indifferent between them, and say that $a$ \emph{weakly prefers} $b$ to $c$.

In an instance of SMTI, a \emph{matching} $M$ is a pairing of men and women such that no one is paired with an unacceptable partner, each man is paired with at most one woman, and each woman is paired with at most one man. In an SRTI instance, a matching is a pairing of agents such that each agent is matched with at most one other agent whom s/he additionally finds acceptable. We write $(a,b) \in M$ to say that $a$ and $b$ are matched in $M$. In an instance of HRT, a matching is a pairing of hospitals and residents such that no agent is paired with an unacceptable candidate, each resident is matched with at most one hospital, and each hospital $h$ is matched with at most $q(h)$ residents. We use $M(a)$ to denote the agent (or the set of agents in the case of hospitals) matched to $a$ in $M$. We write $M(a)=\varnothing$ if agent $a$ is unmatched in $M$. We assume that every agent prefers being matched to an acceptable candidate to remaining unmatched. 

Given an instance of SMTI or SRTI, a matching $M$ is \emph{(weakly) stable} if there is no pair $(a,b) \notin M$ where $a$ prefers $b$ to his current partner in $M$, i.e., $b \succ_a M(a)$, and vice versa. Given an instance of HRT, a matching $M$ is stable if there is no acceptable (resident,hospital) pair $(r,h)$ such that (i) $r$ prefers $h$ to $M(r)$, and (ii) either $|M(h)| < q(h)$ or $h$ prefers $r$ to its worst assigned resident in $M$.

In their seminal work, Gale and Shapley \cite{GS62} showed that every instance of SMI admits a stable matching that can be found in polynomial time by their proposed algorithm (GS). A simple extension of GS can be used to identify stable matchings in instances of SMTI and mechanisms very similar to GS have been used to compute a stable assignment of residents to hospitals in instances of HRT. An instance of SR need not admit a stable matching \cite{GS62}. Irving~\cite{Irv86} provided a polynomial time algorithm that finds a stable matching in an instance of SR, or reports that none exists. Gusfield and Irving~\cite{fGI89} showed that it is straightforward to generalise this algorithm to instances of SRI.

\subsection{Hard Stable Matching Problems}
As explained above, the problem of identifying a stable matching, or showing that none exits, can be solved efficiently for instances of SMTI, HRT and SRI. In contrast to the case for stable marriage, allowing indifference in the stable roommates problem makes the quest for a stable matching a difficult task. Ronn~\cite{Ron90} showed that \textsc{Weak SRT}, the problem of deciding whether a stable matching exists, given an instance of SRT, is \cplxty{NP}-complete.

The importance of stable solutions has been argued and stressed by economists (see e.g.\cite{Rot84,Rot91,RX94}) and finding a stable matching is at the core of many practical applications such as the assignment of residents to hospitals in the National Resident Matching Program (NRMP) in the United States. In many practical applications, however, it is also important to match as many agents as possible, and thus finding a maximum cardinality stable matching (i.e., a stable matching with the largest size amongst all stable matchings) is a crucial issue.

It is known that (in contrast with SMI) an instance of SMTI might admit stable matchings of different sizes, and GS does not necessarily find the largest. \textbf{\textsc{Max SMTI}}, the problem of determining the maximum cardinality stable matching in an instance of SMTI, is known to be \cplxty{NP}-hard \cite{BMM10,IMS08,MIIMM02,OMa07}, even when the input is heavily restricted. 

The concern of computing a maximum size stable matching extends to instances of HRT.  As SMTI can be seen as a special case of HRT in which every hospital has capacity exactly one, the \cplxty{NP}-hardness of \textsc{Max SMTI} implies that \textbf{\textsc{Max HRT}}, the problem of determining the maximum cardinality stable matching in an instance of HRT, is also \cplxty{NP}-hard.  Conversely, we can express any instance of \textsc{Max HRT} as an instance of \textsc{Max SMTI} using a standard cloning argument \cite{Man13,RS90}: for each hospital $h$ in the HRT instance, we create $q(h)$ identical agents in the SMTI instance, each with capacity one.

\cplxty{NP}-hardness of \textsc{Weak SRT} implies that \textbf{\textsc{Max SRT}}, the problem of identifying a maximum cardinality stable matching in an instance of SRT or reporting that none exists, is also \cplxty{NP}-hard. Note that this problem is hard even when all agents find all other agents acceptable.

Depending on the application, one might be willing to tolerate a small degree of instability if that leads to larger matchings. Two different measurements for the degree of instability have been introduced in the literature: the number of blocking pairs and the number of blocking agents. \textbf{{\textsc{Max Size Min BP SMI}} (respectively \textbf{\textsc{Max Size Min BA SMI}})} is the problem of finding a matching, out of all maximum cardinality matchings, which has the minimum number of blocking pairs (respectively minimum number of blocking agents, i.e. agents who belong to at least one blocking pair) in an instance of SMI. Both of these problems are \cplxty{NP}-hard and very difficult to approximate \cite{BMM10}. Note that these two problems are hard even when there are no ties in preference lists.

Since an instance of SR may not admit a stable matching, it is of interest to find a matching with minimum number of blocking pairs. Abraham et al.~\cite{ABM06} showed that \textbf{\textsc{Min BP SR}}, the problem of identifying a matching which has the minimum number of blocking pairs in an instance of SR, is \cplxty{NP}-hard and very hard to approximate. 
Note that this problem is hard even if all agents rank all the other agents in strict order of preference.

\subsection{Parameterised Complexity}\label{sec:par-complexity}

In this paper we are concerned with the \emph{parameterised complexity} of computational problems that are intractable in the classical sense.  Parameterised complexity provides a multivariate framework for the analysis of hard problems: if a problem is known to be $\NP$-hard, so that we expect the running-time of any algorithm to depend exponentially on some aspect of the input, we can seek to restrict this combinatorial explosion to one or more \emph{parameters} of the problem rather than the total input size.  This has the potential to provide an efficient solution to the problem if the parameter(s) in question are much smaller than the total input size.  A parameterised problem with total input size $n$ and parameter $k$ is considered to be tractable if it can be solved by a so-called \emph{FPT algorithm}, an algorithm whose running time is bounded by $f(k)\cdot n^{\mathcal{O}(1)}$,  where $f$ can be any computable function.  Such problems are said to be \emph{fixed parameter tractable}, and belong to the complexity class $\FPT$.  It should be emphasised that, for a problem to be in $\FPT$, the exponent of the polynomial must be independent of the parameter value; problems which satisfy the weaker condition that the running time is polynomial for any constant value of the parameter(s) (so that the degree of the polynomial may depend on the parameters) are said to belong to the class $\XP$.

For further background on the theory of parameterised complexity, we refer the reader to \cite{downeyfellows13,flumgrohe}.

\subsection{The complexity of Integer Programming}

Many of the algorithms we present in this paper make use of an algorithm for \textsc{Integer Linear Programming} in some way.  This problem is formally stated as follows: given an $m \times k$ matrix $A$ and two $m$-dimensional vectors $\mathbf{b}$ and $\mathbf{c}$ (all with coefficients in $\mathbb{Z}$), find a $m$-dimensional vector $\mathbf{x} \in \mathbb{Z}$ which minimizes the scalar product $\mathbf{c}^{T} \cdot \mathbf{x}$, subject to the $m$ linear constraints given by $A \mathbf{x} \leq \mathbf{b}$, or else report that no vector satisfying the constraints exists.  Note that we can easily translate problems in which we wish to maximise rather than minimise the objective function into this form, and also we can express constraints based on linear equalities as a combination of linear inequalities; for simplicity of presentation we will use both of these generalisations when expressing problems as instances of \textsc{Integer Linear Programming}.

While \textsc{Integer Linear Programming} is $\NP$-hard in general, one of the most celebrated results in parameterised complexity is that this problem belongs to $\FPT$ when parameterised by the number of variables.

\begin{theorem}[\cite{paramalgs}, based on \cite{frank87,kannan87,lenstra83}]\label{thm:ILP}
An \textsc{Integer Linear Programming} instance of size $L$ with $k$
variables can be solved using
$$\mathcal{O}(k^{2.5k+o(k)} \cdot \left(L + \log M_x) \log(M_xM_c)\right)$$
arithmetic operations and space polynomial in $L + \log M_x$, where $M_x$ is an upper bound on the absolute value a variable can take in a solution, and $M_c$ is the largest absolute value of a coefficient in the vector $c$.
\end{theorem}

In Section \ref{sec:typed-minbp} we also need to solve instances of \textsc{Integer Quadratic Programming}, a variant of \textsc{Integer Linear Programming} in which the objective function is quadratic.  Formally, given a $k \times k$ integer matrix $Q$, an $m \times k$ integer matrix $A$ and an $m$-dimensional integer vector $\mathbf{b}$, our goal is to find a vector $\mathbf{x} \in \mathbb{Z}^k$ which minimises $\mathbf{x}^T Q \mathbf{x}$, subject to the $m$ linear constraints $A \mathbf{x} \leq b$, or else report that no vector satisfying the constraints exists.  As before, we note that we can easily generalise this definition to deal with maximisation problems and constraints in the form of linear equalities.  Lokshtanov recently gave an FPT algorithm for this problem.

\begin{theorem}[\cite{lokshtanov15}]\label{thm:IQP}
\textsc{Integer Quadratic Programming} is in $\FPT$ parameterised by $k + \alpha$, where $\alpha$ is the maximum absolute value of any entry in the matrices $A$ and $Q$.
\end{theorem}

\subsection{Related Work}\label{sec:related-work}

\paragraph{\cplxty{NP}-hard matching problems.} 
The \cplxty{NP}-hardness of \textsc{Max SMTI} has been shown for a variety of restricted 
for example: (1) even if each man's list is strictly ordered, and each woman's list is either strictly ordered or is a tie of length 2 \cite{MIIMM02}, (2) even if each mans preference list is derived from a strictly-ordered master list of women, and each woman's preference list is derived from a master list of men that contains only one tie \cite{IMS08}, and (3) even if the SMTI instance has symmetric preferences; that is, for any acceptable (man, woman) pair $(m_i,w_j)$, $rank(m_i,w_j) = rank(w_j ,m_i)$ \cite{OMa07}, where $rank(a,b)$ is defined to be one plus the number of candidates that $a$ prefers to $b$. As SMTI is a special case of HRT, the \cplxty{NP}-hardness of finding a maximum stable matching in the latter follows directly from the \cplxty{NP}-hardness of this problem in the former. 

The \cplxty{NP}-completeness of \textsc{Weak SRT}, and hence \textsc{Max SRT}, holds even if each preference list is either strictly ordered or contains a tie of length 2 at the head. 

\textsc{Max Size Min BP SMI} and \textsc{Max Size Min BA SMTI} are \cplxty{NP}-hard and very hard to approximate \cite{BMM08} even if each agent's preference list is of length at most 3 \cite{BMM10,HIM09}, but polynomial-time solvable if agents on one side of the market have preference lists of length at most 2 \cite{BMM10}.

\paragraph{Parameterized complexity of matching problems.}
There are a limited number of works addressing fixed-parameter tractability in stable matching problems. 
Marx and Schlotter 
\cite{MS10} gave the first parameterised complexity results on \textsc{Max SMTI}. They show that the problem is in \cplxty{FPT} when parameterised by the total length of the ties, but is \cplxty{W[1]}-hard when parameterised by the number of ties in the instance, even if all the men have strictly ordered preference lists. Very recently, three different works have studied hard stable matching problems from the perspective of parameterised complexity. In \cite{MS17}, the authors obtained results on the parameterised complexity of finding a stable matching which matches a given set of distinguished agents and has as few blocking pairs as possible. In \cite{GSZ17} it is shown that several hard stable matching problems, including \textsc{Max SMTI}, are \cplxty{W[1]}-hard when parameterised by the treewidth of the graph obtained by adding an edge between each pair of agents that find each other mutually acceptable.  In \cite{GRSZ17}, the authors study above guarantee parameterisations of the problem of finding a stable matching that balances the dissatisfaction of men and women, with parameters that capture the degree of dissatisfaction.

\paragraph{Attributes and types.} 
Settings in which agents are partitioned into different types, or derive their preferences based on a set of attributes assigned to each candidate, have been considered for the problems of sampling and counting stable matchings in instances of SM or SR (see, e.g., \cite{BGR08,CGM12,cgm12a}). In \cite{EFSY13}, the authors study the problem of characterising matchings that are rationalisable as stable matchings when agents' preferences are unobserved. They focus on a restricted setting that translates into assigning each agent a type based on several attributes, and assuming that agents of the same type are identical and have identical preferences. They remark that empirical studies on marriage typically make such an assumption \cite{CS06}. Bounded agent types have been considered in \cite{AK-11,SAK-10} to derive polynomial-time results for the coalition structure generation problem, an important issue in cooperative games when the goal is to partition the participants into exhaustive and disjoint coalitions in order to maximise the social welfare.

\section{Our basic model: agents of the same type are indistinguishable}\label{sec:typed}
In this section we begin with a formal definition of the simplest model we consider, in which agents' preferences can be derived directly from the preferences of types over types of candidates. This implies that agents of the same type are completely indistinguishable. That is, they have the same preference lists, and every other agent that finds their type acceptable is indifferent between them. 

We then identify a necessary and sufficient condition, in terms of the types of the least and the second least desirable partners assigned to any agent of each type, for a matching to be stable in this model.  We use this to show that, if there are $k$ types, we can solve \textsc{Max SRTI} by solving $k^{\mathcal{O}(k)}$ instances of \textsc{Integer Linear Programming}. 
This implies that \textsc{Max SRTI} and hence \textsc{Max SMTI}, parameterised by $k$, belong to \cplxty{FPT}.
We then show that we can solve \textsc{Max SMTI} more efficiently by solving 
$k^{\mathcal{O}(k)}\cdot\log n$ instances of \textsc{Max Flow} on directed networks with $\mathcal{O}(k)$ vertices and maximum edge capacity $\mathcal{O}(n)$. This alternative approach reduces the time complexity of \textsc{Max SMTI} from $k^{\mathcal{O}(k^2)}\log^3 n+\mathcal{O}(n)$ to $k^{\mathcal{O}(k)} \cdot \log^2 n+\mathcal{O}(n)$. The $\mathcal{O}(n)$ part of the time complexities accounts for the time required to construct the maximum cardinality matching.

Lastly, we extend the methods for solving \textsc{Max SRTI} to provide FPT algorithms for \textsc{Max Size Min BP SMTI} and \textsc{Max Size Min BA SMTI}.

\subsection{Definition of typed instances}\label{sec:def-basic-model}
Assume that there are $k$ types available for agents. Let $[k]$ denote the set $\{1,2,\ldots,k\}$.
Let $\agentset_i$ denote the set of agents that are of type $i$. Thus we have that the set of agents $\agentset=\bigcup_{i\in[k]} \agentset_i$.
Each type $i$ has a preference ordering over types of the candidates, which need not be complete or strict. 
We assume, without loss of generality, that $|\agentset_i| >0 $ for all $i\in[k]$, and that each type finds at least one other type acceptable.
We write $j \succ_i \ell$, or equivalently $\ell \prec_i j$, if agents of type $i$ strictly prefer agents of type $j$ to agents of type $\ell$. 
We write $j \tie_i \ell$ to denote that agents of type $i$ are indifferent between agents of types $j$ and $\ell$, and $j \succeq_i \ell$ if agents of type $i$ prefer agents of type $j$ to those of type $\ell$ or are indifferent between the two. We assume that given every two agents $x$ and $y$ of the same type:
\begin{enumerate}
\item $x$ and $y$ have identical preference lists when restricted to $\agentset \setminus \{x,y\}$
, and
\item all other agents are indifferent between $x$ and $y$.
\end{enumerate}

These requirements imply that any agent either finds all agents of a given type acceptable (and is indifferent between them) or finds none of them acceptable.  We say that an instance of a stable matching problem satisfying these requirements is \emph{typed}, and refer to the standard problems with input of this form as \textsc{Typed Max SMTI} etc.  Note that \textsc{Typed Max SMTI}, and hence \textsc{Typed Max SRTI}, remain \cplxty{NP}-hard when $k$ is considered to be part of the input: we can always create a typed instance by assigning each agent its own type. 

A typed instance $I$ of SRTI is given as input by specifying the number of types $k$ and, for each type $i$, the set $\agentset_i$ of agents of type $i$ as well as the preference ordering $\succ_i$ over types of the candidates.
Observe that, if we are only given the preference list for each agent as input, it is straightforward to compute, in polynomial time, the coarsest partition of the agents into types that satisfies the definition of a typed instance.  Having found such a partition, the preference lists over types can also be constructed efficiently.

\begin{example}\label{ex:model1}
Assume we have 4 types for the agents in a stable marriage setting, and that all men are of type $1$ and types $2$, $3$ and $4$ correspond to women. Let the preference ordering of type $1$ over types of women be as follows, where the preference list is ordered from left to right in decreasing order of preference, and the types in round brackets are tied: $(2 ~ 3) ~ 4$. Assume that there are 7 women and $w_1$ and $w_2$ are of type $2$, $w_3$ and $w_4$ are of type $3$, and $w_5$, $w_6$ and $w_7$ are of type $4$. Therefore, the preference lists of all men under the typed model are as follows: $(w_1 ~ w_2 ~ w_3 ~ w_4) ~ (w_5 ~ w_6 ~ w_7)$. 
\end{example}

\subsection{An FPT algorithm for \textsc{Typed Max SRTI}}\label{sec:typed-maxsrti}

Let $I$ be a typed instance of SRTI, and let $M$ be a matching in $I$. 
We may assume without loss of generality that every agent is matched, by creating sufficiently many dummy agents of type $k+1$ which are inserted at the end of each agent's (possibly incomplete) preference list.
We define $\worst_M(i)$ and $\secondworst_M(i)$ to be the types of the least desirable agent and the second least desirable agent with which any agent of type $i$ is matched in $M$, respectively,
breaking ties arbitrarily (e.g. lexicographically). 
Note that $\worst_M(i)$ would be the dummy type if an agent of type $i$ is unmatched (i.e. matched to a dummy agent) in $M$. 
If there is only one agent of type $i$, then $\secondworst_M(i)$ is undefined, in which case we let $\secondworst_M(i)=\varnothing$.
Note that it is possible to have $\secondworst_M(i)=\worst_M(i)$.
Let $\type(a)$ denote the type of a given agent $a$.

The key observation is that, in order to determine whether or not $M$ is stable, it suffices to examine the values of $\worst_M(i)$ and $\secondworst_M(i)$ for each $i \in [k]$.

\begin{lemma}\label{stability-test-srti}
Let $I$ be a typed instance of SRTI. Then a matching $M$ in $I$ is stable if and only if (1) there is no pair $(i,j) \in [k]^{(2)}$, $i\neq j$, such that $j \pref_i \worst_M(i)$ and $i \pref_j \worst(j)_M$, and (2) there is no pair $(i,i)$, $i \in [k]$, such that there are at least two agents of type $i$ and $i \pref_i \secondworst_M(i)$.
\end{lemma}
\begin{proof}
Suppose first that $M$ is not stable.  In this case, by definition, there exists some pair of agents $(a,b)$ such that $a$ and $b$ are not matched together but each prefers the other over their current partner. First suppose that $a$ and $b$ are of different types and without without loss of generality assume that $a$ is of type $i$ and $b$ is of type $j$, $i\neq j$.  Then we know that $j \pref_i \type(M(a)) \weaklypref_i \worst_M(i)$, and similarly $i \pref_j \type(M(b)) \weaklypref_j \worst_M(j)$. Now suppose that $a$ and $b$ are both of the same type $i$. Then we know that type $i$ likes type $\type(M(a))$ or type $\type(M(b))$ at least as well as type $\secondworst_M(i)$. Without loss of generality assume that $\type(M(a)) \weaklypref_i \secondworst_M(i)$. Since $(a,b)$ is a blocking pair, $\type(b) \pref_i \type(M(a))$, and therefore $i \pref_i \secondworst_M(i)$.

Conversely, suppose that $M$ is stable.  Suppose for a contradiction that at least one of the two conditions in the statement of the theorem does not hold. First suppose that the first condition does not hold, i.e. there is a pair of types $(i,j)$, $i\neq j$, such that $j \pref_i \worst_M(i)$ and $i \pref_j \worst_M(j)$. Then there is some agent $a$ of type $i$ which is matched with an agent of type $\worst_M(i)$, so in particular $a$ is matched with a partner less desirable than any agent of type $j$.  Similarly, there is some agent $b$ of type $j$ which is matched with an agent of type $\worst_M(j)$ and hence is matched with a partner less desirable than any agent of type $i$.  Thus $a$ and $b$ both prefer each other to their current partner, and so form a blocking pair. This contradicts the assumption that $M$ is stable. Now suppose that the second condition does not hold, i.e. there exists a type $i$ corresponding to at least two agents where $i \pref_i \secondworst_M(i)$. Then there are two agents $a$ and $b$ of type $i$ who are matched to an agent of type $\worst_M(i)$ and an agent of type $\secondworst_M(i)$, respectively. Thus $a$ and $b$ both prefer each other to their current partner, and so form a blocking pair. This contradicts the assumption that $M$ is weakly stable.
\end{proof}

We say that a matching $M$ \emph{realises} given functions $\worst: [k] \rightarrow [k+1]$ and $\secondworst: [k] \rightarrow [k+2]$ if, for each $i \in [k]$, the least desirable and the second least desirable partners any agent of type $i$ has in $M$ are of types no worse than $\worst(i)$ and $\secondworst(i)$ respectively. 
We say that a pair of functions $(\worst, \secondworst)$ is feasible if for each type $i$, (1) $\worst(i)$ is either a type acceptable to type $i$ or the dummy type, (2) $\secondworst(i)$ is either a type acceptable to type $i$ or the dummy type, or $\varnothing$ (only if there is only one agent of type $i$), and (3) if $\worst(i)=\secondworst(i)$ then there exist at least two agents of type $\worst(i)$.
We say that a pair of functions $(\worst, \secondworst)$ is
\emph{$I$-stable} for an instance $I$ of SRTI if it is feasible, there is no pair 
$(i,j) \in [k]^{(2)}$
such that $j \pref_i \worst(i)$ and $i \pref_j \worst(j)$, and there is no pair $(i,i)$, $i \in [k]$, such that there are at least two agents of type $i$ and $i \pref_i \secondworst_M(i)$.

Given a pair of $I$-stable functions $(\worst, \secondworst)$, we write $\max(\worst,\secondworst)$ for the maximum cardinality of any matching in $I$ that realises $\worst$ and $\secondworst$. Using Lemma \ref{stability-test-srti}, it is straightforward to check that, given a typed instance $I$ of SRTI, the cardinality of a solution to \textsc{Max SRTI} can be found by taking the largest value of $\max(\worst,\secondworst)$ over all pairs of $I$-stable functions $(\worst, \secondworst)$.

\begin{corollary}\label{lma:max-I-stable}
Let $I$ be a typed instance of SRTI.  Then the cardinality of the largest stable matching in $I$ is equal to 
$$\max\{\max(\worst,\secondworst): (\worst,\secondworst) \textup{ is $I$-stable}\}.$$
\end{corollary}

We next show that, given a pair of $I$-stable functions $(\worst,\secondworst)$, there is an FPT algorithm (with parameter $k$) to compute $\max(\worst,\secondworst)$.

\begin{lemma}\label{lma:max-srti-calculate}
Let $I$ be a typed instance of $SRTI$, and fix a pair of $I$-stable function $(\worst, \secondworst)$.  We can compute $\max(\worst,\secondworst)$ in time $2^{\mathcal{O}(k^2)}\log^3 n$.
\end{lemma}
\begin{proof}
Suppose that in total there are $n$ agents. We compute $\max(\worst,\secondworst)$ by solving a suitable instance of an Integer Linear Programming. For each unordered pair of distinct values $\{i,j\} \in [k+1]^{(2)}$, the variable $n_{\{i,j\}}$ represents the number of pairs in the matching consisting of one agent of type $i$ and another of type $j$. Recall that $N_i$ is the set of agents of type $i$. We then have the following integer linear program:

\begin{align*}
\text{\textbf{maximize}} &~ ~ \sum_{1 \leq i, j \leq k} n_{\{i,j\}} & \\
\text{\textbf{subject to}} &~ ~ \sum_{j\in[k+1]} n_{\{i,j\}} = |N_i|, & ~ \forall i \in [k] \\
 ~ &~ ~ n_{\{i,j\}} \geq 0, & ~ \forall i,j \in [k+1] \\
  ~ &~ ~ \sum_{j \succeq_i \worst(i)} n_{\{i,j\}} = |N_i|, & ~ \forall i \in [k] \\
\text{\textbf{and}} &~ ~ \sum_{j \weaklypref_i \secondworst(i)} n_{\{i,j\}} \geq |N_i| - 1, & ~ \forall i \in [k].
\end{align*}

The first two constraints ensure that every agent is involved in exactly one pair, perhaps with a dummy agent.  The third constraint ensures that no agent of type $i$ is assigned a partner of type worse than $\worst(i)$, and the last constraint ensures that no more than one agent of type $i$ is assigned a partner of type worse than $\secondworst(i)$.
The objective function seeks to maximise the total number of pairs that do not involve dummy agents.

The above integer linear program has $3k+(k+1)^2$ constraints and $(k+1)^2$ variables. The upper bound on the absolute value a variable can take is $n$. Therefore, by Theorem \ref{thm:ILP}, this maximisation problem for any pair of candidate functions $(\worst,\secondworst)$ can be solved in time $k^{\mathcal{O}(k^2)}\log^3 n$.
\end{proof}

It then follows that \textsc{Typed Max SRTI} is in \cplxty{FPT} parameterised by the number $k$ of different types in the instance.

\begin{corollary}\label{cor:typed-maxsrti-fpt}
\textsc{Typed Max SRTI} can be solved in time $k^{\mathcal{O}(k^2)}\log^3 n + \mathcal{O}(n)$. If we are only interested in computing the size of the maximum cardinality matching, and not the matching itself, this can be done in time $k^{\mathcal{O}(k^2)}\log^3 n$.
\end{corollary}
\begin{proof} Let $I$ be a typed instance of SRTI. We consider each possible pair of functions $\worst: [k] \rightarrow [k+1]$ and $\secondworst: [k] \rightarrow [k+2]$ in turn; there are at most $(k+1)^k(k+2)^k$ such pairs of functions.
By Lemma \ref{stability-test-srti}, we can determine in time $\mathcal{O}(k^2)$ whether $(\worst,\secondworst)$ is $I$-stable. 
If none of the candidate pairs of functions is $I$-stable, we conclude that there is no stable matching and terminate.
For each $I$-stable pair of functions $(\worst, \secondworst)$, we compute $\max(\worst,\secondworst)$ in time $k^{\mathcal{O}(k^2)}\log^3 n$, by Lemma \ref{lma:max-srti-calculate}. We then take the maximum value of $\max(\worst, \secondworst)$ over all $I$-stable pairs of functions $(\worst,\secondworst)$ which, by Corollary \ref{lma:max-I-stable}, is equal to the cardinality of the largest stable matching in $I$. Given all values of $n_{\{i,j\}}$ and using suitable data structures, a maximum cardinality matching can easily be constructed in time $\mathcal{O}(n)$.
\end{proof}

\subsection{An FPT algorithm for \textsc{Typed Max SMTI}}\label{sec:typed-maxsmti}

SMTI is a bipartite restriction of SRTI and therefore, by Corollary \ref{cor:typed-maxsrti-fpt}, we see that \textsc{Typed Max SMTI} belongs to FPT when parameterised by the number of types in the instance.  In this section we show that the running time can be improved in the bipartite setting of SMTI and HRT.  

In this setting, we no longer need to consider the function $\secondworst$ in order to determine stability or otherwise of the matching: it is enough to consider only the type of the least desirable partner assigned to any agent of each type.  The following characterisation of stability in terms of the function $\worst$ can easily be deduced from the proof of Lemma \ref{stability-test-srti}.

\begin{lemma}
Let $I$ be a typed instance of SMTI. Then a matching $M$ in $I$ is stable if and only if there is no pair $(i,j) \in [k]^{(2)}$, $i\neq j$, such that $j \pref_i \worst_M(i)$ and $i \pref_j \worst(j)_M$.
\end{lemma}

We also extend the terminology of the previous section to describe a function $\worst$ as $I$-stable if there is no pair $(i,j) \in [k]^{(2)}$, $i\neq j$, such that $j \pref_i \worst(i)$ and $i \pref_j \worst(j)$.  Similarly, wee say that a matching $M$ realises $\worst$ if the least desirable partner of any agent of type $i$ is no worse than $\worst(i)$, and write $\max(\worst)$ for the cardinality of the largest matching which realises $\worst$.  With this terminology, it is clear that (as in Corollary \ref{lma:max-I-stable}) the cardinality of the largest stable matching in a typed instance $I$ of SMTI is equal to $\max\{\max(\worst): \worst \text{is $I$-stable}\}$.

These observations alone would allow us to simplify the ILP formulation by omitting the final constraint, and would reduce the number of instances of ILP that must be solved to $k^k$.  However, we can make further improvements by using a network flow method to compute the value of $\max(\worst)$ for each relevant function $\worst$ in polynomial time.

\begin{lemma}\label{lma:dec-flows}
Let $I$ be a typed instance of $SMTI$, and fix an $I$-stable function $\worst$.  We can compute $\max(\worst)$ in time $\mathcal{O}(k^3 \log^2 n)$.
\end{lemma}
\begin{proof}
The proof is structured as follows. 
Suppose that in total there are $n_1$ women and $n_2$ men, so we have that $n=n_1+n_2$.
Note that $\max(\worst)$ is at most $\min\{n_1, n_2\}$, which in turn is at most $\lfloor n/2 \rfloor$. Therefore, using a binary search strategy, we can determine the maximum size of a matching realising $\worst$ by solving $\mathcal{O}(\log n)$ instances of the decision problem ``Is $\max(\worst)$ at least $c$?", where $c \in \{1,\ldots,\min\{n_1, n_2\}\}$.
We will show that we can determine whether $\max(\worst) \geq c$ by solving \textsc{Max Flow} on a directed network $D$ with $\mathcal{O}(k)$ vertices, in which the maximum capacity of any edge is $\mathcal{O}(n)$ (see Figure~\ref{fig:flow}); we can construct $D$ from $I$ in time $\mathcal{O}(k^2 \log n)$. \textsc{Max Flow} can be solved on $D$ in time $\mathcal{O}(k^3 \log n)$, using an algorithm due to Orlin \cite{orlin13}, where the $\log n$ factor is required to carry out arithmetic operations on integers of size $\mathcal{O}(n)$. Therefore, we conclude that we can compute $\max(\worst)$ in time $\mathcal{O}(k^3 \log^2 n)$.

We now show how to construct the network $D$ (depicted in Figure~\ref{fig:flow}).
Assume, without loss of generality, that types $1,\ldots,k_1$ are types of women and types $k_1+1,\ldots,k$ are types of man. 
We construct a directed network $D=(V,E)$ with vertex set $V= \{s,t,v_1,\ldots,v_k,d_w,d_m\}$, where vertex $v_i$ corresponds to type $i$, and $d_w$ and $d_m$ both correspond to the dummy type. The edge set $E$ is consist of the following directed edges with associated capacities:
\begin{itemize}
\item $\overrightarrow{sv_i}$ for $1 \leq i \leq k_1$, with capacity $|N_i|$,
\item $\overrightarrow{sd_w}$, with capacity $n_2-c$,
\item $\overrightarrow{v_iv_j}$ for all pairs $(i,j)$ with $1 \leq i \leq k_1$ and $k_1 + 1 \leq j \leq k$ such that $j \weaklypref_i \worst(i)$ and $i \weaklypref_j \worst(j)$, with capacity $\min\{|N_i|,|N_j|\}$,
\item $\overrightarrow{v_id_m}$ for all $1 \leq i \leq k_1$ such that $\worst(i) = k+1$, with capacity $|N_i|$,
\item $\overrightarrow{d_wv_j}$ for all $k_1 + 1 \leq j \leq k$ such that $\worst(j) = k + 1$, with capacity $|N_j|$,
\item $\overrightarrow{d_wd_m}$, with capacity $\min \{n_1,n_2\} - c$,
\item $\overrightarrow{v_jt}$ for $k_1 + 1 \leq j \leq k$, with capacity $|N_j|$,
\item $\overrightarrow{d_mt}$, with capacity $n_1-c$.
\end{itemize}

\begin{figure}[h]
\centering
	\includegraphics[width=0.7\textwidth]{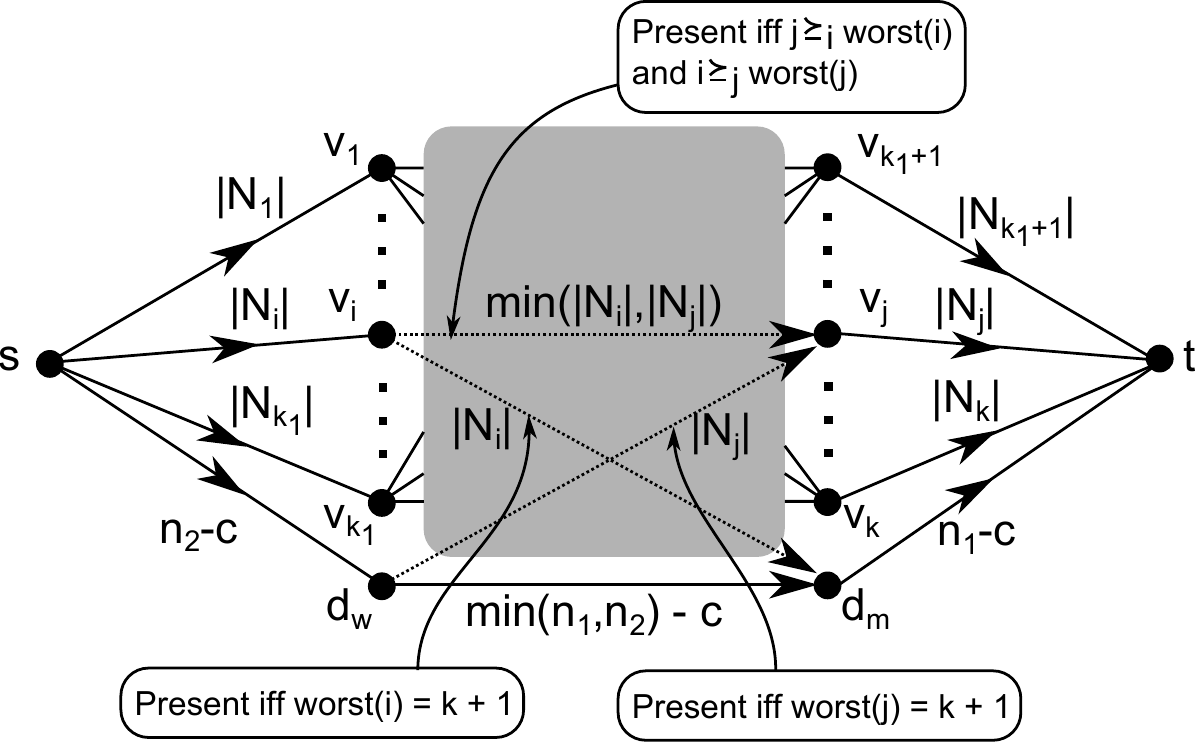}
	\caption{Network $D$, constructed from an instance of SMTI in the proof of Lemma~\ref{lma:dec-flows}. Types $1,\ldots,k_1$ are types of women and types $k_1+1,\ldots,k$ are types of man. Vertices $d_w$ and $d_m$ both correspond to the dummy type. 
	}\label{fig:flow}
\end{figure}

By construction, there are $\mathcal{O}(k)$ vertices in the network and the maximum capacity of any edge is $\mathcal{O}(n)$. Therefore, it is clear that we can construct $D$ from $I$ in time $\mathcal{O}(k^2 \log n)$ where the $\log n$ factor allows for the time required to read integers of size $\mathcal{O}(n)$ in the input. So it only remains to show that we can decide whether $\max(\worst) \geq c$ by solving \textsc{Max Flow} on $D$.

We claim that the maximum flow from $s$ to $t$ in $D$ is equal to $n_1 + n_2 - c$ if and only if there is a matching of cardinality at least $c$ which realises $\worst$ (i.e., if and only if $\max(\worst) \geq c$). Note that the flow clearly cannot exceed this value, summing the capacities of the outgoing edges of $s$. 

$(\Rightarrow)$: Suppose first that there exists a matching $M$ of cardinality at least $c$ that realises $\worst$. We define a flow $f$ in $D$ of value $n_1 + n_2 - c$ as follows:
\begin{itemize}
\item for all edges incident with $s$ or $t$, we assign a flow equal to the capacity;
\item for each edge $\overrightarrow{v_iv_j}$, we assign a flow equal to the number of pairs in $M$ consisting of one agent of type $i$ and one agent of type $j$;
\item for each edge $\overrightarrow{v_id_m}$, we assign a flow equal to the number of unassigned agents of type $i$ in $M$;
\item for each edge $\overrightarrow{d_wv_j}$, we assign a flow equal to the number of unassigned agents of type $j$ in $M$; and
\item to $\overrightarrow{d_wd_m}$, we assign a flow equal to $|M| - c$.
\end{itemize}
We need to show that $f$ is feasible and that value of $f$ is $n_1 + n_2 - c$.
It is easy to verify that no edge is assigned a flow greater than its capacity, so 
to prove that $f$ is feasible 
it suffices to demonstrate that the flow is conserved at each vertex other than $s$ and $t$.  For any vertex $v_i$ with $1 \leq i \leq k_1$, the flow into $v_i$ is equal to $|N_i|$, and the total flow out of $v_i$ is the sum of the numbers of agents of type $i$ matched with agents of type $j$ (for all $j\geq k_1+1$) plus the number of agents of type $i$ that are unmatched, which must be equal to $|N_i|$. A symmetric argument holds for each vertex $v_j$ with $k_1+1 \leq j \leq k$.  The flow into $d_w$ is $n_2 - c$, and the total flow out of $d_w$ is equal to the total number of men that are unmatched in $M$, plus $|M| - c$: since each pair in $M$ contains exactly one distinct man, the number of unmatched men in $M$ is equal to $n_2 - |M|$, so the total flow out of $d_w$ is $n_2 - |M| + (|M| - c) = n_2 - c$.  A symmetric argument for $d_m$ completes our claim that we have defined a feasible flow. It remains to show that $f$ does indeed have value $n_1 + n_2 - c$: note that every outgoing edge from $s$ is saturated, and the sum of the capacities of these edges is equal to $n_1 + n_2 - c$.

$(\Leftarrow)$: Conversely, suppose that we have a flow $f$ in $D$ with value $n_1 + n_2 - c$. 
We define a matching $M_f$ in $I$ by 
greedily including $f(\overrightarrow{v_iv_j})$ pairs consisting of one agent of type $i$ and one agent of type $j$, for every edge $\overrightarrow{v_iv_j}$ in $D$. We need to show that $M_f$: (1) is a matching, (2) has cardinality at least $c$, and (3) realises $\worst$. To check that $M_f$ is a matching, we need to verify that we have not assigned more than $|N_i|$ agents of any type $i$.  If $i$ is a type of women, then the total number of pairs involving an agent of type $i$ is the total flow out from $v_i$, which cannot be more than the capacity of the single incoming edge to $v_i$, which is equal to $|N_i|$. Similarly, if $i$ is a type of men, then the total number of pairs involving an agent of type $i$ is the total flow into $v_i$, which cannot be more than the capacity of the single outgoing edge from $v_i$, which is equal to $|N_i|$. 
The cardinality of $M_f$ is equal to the sum of flows on edges $\overrightarrow{v_iv_j}$. As the set of such edges, together with all outgoing edges of $d_w$ and all incoming edges of $d_m$, forms an $s$-$t$ cut in $D$ (in which all edges are directed the same way), we can conclude that the cardinality of $M_f$ is equal to the value of $f$ minus the sum of flows on outgoing edges of $d_w$ and incoming edges of $d_m$. Note that the sum of flows on outgoing edges of $d_w$ cannot be more than the capacity of the single incoming edge of $d_w$, which is $n_2 - c$, and similarly the sum of flows on incoming edges of $d_m$ cannot be more than $n_1 - c$.  Thus we see that $|M_f|$ is at least $n_1 + n_2 - c - (n_2 - c) - (n_1 - c) = c$.
It is clear by construction of $D$ and $M_f$ that $M_f$ realises $\worst$. This completes our proof.
\end{proof}

The next corollary then follows along the same lines as Corollary \ref{cor:typed-maxsrti-fpt}.

\begin{corollary}\label{cor:typed-maxsmti-fpt}
Computing the size of the maximum cardinality matching in a typed instance of SMTI can be done in time $k^{\mathcal{O}(k)} \cdot \log^2 n$.
\end{corollary}

\subsection{FPT algorithms for \textsc{Typed Max Size Min BP SMTI} and \textsc{Typed Max Size Min BA SMTI}}
\label{sec:typed-minbp}

As stated in Section \ref{sec:introduction}, in some settings the size of the matching takes priority over the stability criterion. That is, the mechanism designers are willing to tolerate a small degree of instability if that leads to a matching of larger size. We can extend the methods from previous sections to deal with this situation in both the bipartite and non-bipartite setting; for simplicity of presentation we begin with the case of SMTI and then discuss how to extend the method to SRTI.

We begin by considering the problem of minimising the total number of blocking pairs. For the rest of this section, we assume that types $1$ to $k'$ are types of women, and types $k'+1$ to $k$ are types of men.  We further assume that all preference lists are extended to include the dummy type as their least desirable acceptable type. If an agent $a$ is unmatched in $M$, we say that $M(a)$ is of type $k+1$, the dummy type. As a first step, we translate the definition of a blocking pair 
into the setting of typed instances.

\begin{prop}\label{prop:bp-mod1}
Let $x$ and $y$ be agents of type $i$ and type $j$ respectively, and suppose that $M(x)$ is of type $j'$ 
and $M(y)$ is of type $i'$. 
Then $(x,y)$ is a blocking pair in $M$ if and only if $j \succ_i j'$ and $i \succ_j i'$.
\end{prop}

Given this observation, we can obtain an expression for the total number of blocking pairs in $M$ 
that are comprised of an agent of type $i$ and an agent of type $j$.

\begin{lemma}
The number of blocking pairs $(x,y)$ in $M$ such that $x$ is a woman of type $i$ and $y$ is a man of type $j$ is given by
$$\left(\sum_{i' \prec_j i} n_{\{i',j\}} \right) \left(\sum_{j' \prec_i j} n_{\{i,j'\}} \right) = \sum_{1 \leq i', j' \leq k+1} n_{\{i',j\}} n_{\{i,j'\}} \left( \mathbb{1}_{i' \prec_j i} \mathbb{1}_{j' \prec_i j}\right),$$
where $\mathbb{1}_{a \prec_b c}$ is an indicator function that returns one if $a \prec_b c$ and zero otherwise. 
\end{lemma}
\begin{proof}
It is easy to verify that the left-hand side and the right-hand side of the equation are equal. For the remainder of the proof, we focus on the left-hand side of the equation. 
By Proposition \ref{prop:bp-mod1}, we know that the set of blocking pairs consisting of one agent of type $i$ and another of type $j$ is precisely
$$\{(x,y): \type(M(x)) \prec_i t_j \text{ and } \type(M(y)) \prec_j t_i\}.$$

The cardinality of this set is thus equal to the number of agents of type $i$ that are matched to an agent of type inferior to type $j$ ($\sum_{j' \prec_i j} n_{\{i,j'\}}$) multiplied by the number of agents of type $j$ that are matched to an agent of type inferior to type $i$ ($\sum_{i' \prec_j i} n_{\{i',j\}}$).
The result follows immediately.
\end{proof}

Summing over all possibilities for $i$ and $j$ gives the following result.

\begin{lemma}
The total number of blocking pairs in $M$ is given by
$$\sum_{1 \leq i,j \leq k} \sum_{1 \leq i', j' \leq k+1} n_{\{i',j\}} n_{\{i,j'\}} \left( \mathbb{1}_{i' \prec_j i} \mathbb{1}_{j' \prec_i j} \right).$$
\end{lemma}

We can now prove our main result concerning \textsc{Typed Max Size Min BP SMTI}.

\begin{theorem}\label{thm:min-bp-fpt}
\textsc{Typed Max Size Min BP SMTI} belongs to $\FPT$ when parameterised by the number $k$ of different types in the given instance.
\end{theorem}
\begin{proof}
We begin by computing, in polynomial time, the cardinality $C_{\max}$ of  a maximum matching in our instance.  Our strategy then is to formulate \textsc{Max Size Min BP SMTI} as an instance of \textsc{Integer Quadratic Programming}.  Our goal is to minimise the following objective function
\begin{align*}
\sum_{1 \leq i,j \leq k} \sum_{1 \leq i', j' \leq k+1} & n_{\{i',j\}} n_{\{i,j'\}} \left( \mathbb{1}_{i' \prec_j i} \mathbb{1}_{j' \prec_i j} \right) \\
	=& \sum_{\substack{1 \leq i,i' \leq k' \\ k'< j,j' \leq k}} \left( \mathbb{1}_{i' \prec_j i} \mathbb{1}_{j' \prec_i j} + \mathbb{1}_{i \prec_{j'} i'} \mathbb{1}_{j \prec_{i'} j'} \right) \\
    & \qquad + \sum_{\substack{1 \leq i \leq k' \\ k' \leq j,j' \leq k}} n_{\{k+1,j\}}n_{\{i,j'\}} \mathbb{1}_{j' \prec_{i} j} \\
    & \qquad + \sum_{\substack{1 \leq i,i' \leq k' \\ k' \leq j \leq k}} n_{\{i',j\}}n_{\{i,k+1\}} \mathbb{1}_{i' \prec_j i} \\ 
    & \qquad + \sum_{\substack{1 \leq i \leq k' \\ k' < j \leq k}} n_{\{k+1,j\}} + n_{\{i,k+1\}},
\end{align*}
subject to the constraints
\begin{align*}
~ & \sum_{j \in [k+1]} n_{\{i,j\}} = |N_i| & ~ \forall i \in [k]\\
\text{\textbf{and}} & ~ ~ \sum_{1 \leq i, j \leq k} n_{\{i,j\}} = C_{\max}.
\end{align*}

To see that the right-hand side of the objective function equation is equal to the left-hand side of the equation, notice that every pair $(\{i,j\},\{i',j'\}) \in [k]^{(2)}$ appears twice in the summation on the left-hand side: once with $\{i,j\}$ coming from the first sum and $\{i',j'\}$ coming from the second sum, and once the other way around. The former counts the number of blocking pairs $(x,y)$ where $x$ is of type $i$, $y$ is of type $j$, and $M(x)$ and $M(y)$ are of types $j'$ and $i'$ respectively. The latter counts the number of blocking pairs $(x,y)$ where $x$ is of type $i'$, $y$ is of type $j'$, $M(x)$ is of type $j$, and $M(y)$ is of type $i$.  Pairs $(\{i,j\},\{i',j'\})$ where at least one of $i,i',j,j'$ is equal to the dummy type $k+1$ are dealt with separately (in this case there can be at most one blocking pair).

The linear constraints enforce that every agent is involved in exactly one pair (perhaps with a dummy agent), and that the number of pairs that do not involve dummy agents is equal to the maximum possible cardinality of a matching.  
We can write our objective function in the form $x^T Q x$ where $x$ is the vector $(n_{\{1,1\}},$ $ n_{\{1,2\}},$ $\ldots,$ $n_{\{k,k\}})^T$ and the entry of $Q$
corresponding to $n_{\{i,j'\}}$ and $n_{\{i',j\}}$ is equal to either $0$, $1$ or $2$ depending on how many of the following conditions hold: 
 \begin{enumerate}
 \item $j \succ_i j'$ and $i \succ_{j} i'$, and
 \item $j' \succ_{i'} j$ and $i' \succ_{j'} i$.
 \end{enumerate} 
Thus, by Theorem \ref{thm:IQP}, we have an FPT algorithm to solve our instance of \textsc{Max Size Min BP SMTI}.
\end{proof}

We now consider the problem of minimising the number of agents which are involved in at least one blocking pair.  We start by characterising the conditions under which an agent of a particular type can belong to one or more blocking pairs; this characterisation follows immediately from the definition of a blocking pair. 

\begin{lemma}\label{lma:num-ba}
Let $x$ be an agent of type $i$ and assume that $M(x)$ is of type $j$ (which would be a dummy type if $x$ is unmatched).  Then $x$ belongs to a blocking pair if and only if there is some agent $y$ of type $j'$ who is paired with an agent of type $i'$ (which would be a dummy type if $y$ is unmatched) such that $i \succ_{j'} i'$ and $j' \succ_i j$.
\end{lemma}

Using this characterisation, we can now prove our main result concerning \textsc{Typed Max Size Min BA SMTI}.

\begin{theorem}\label{thm:min-ba-fpt}
\textsc{Typed Max Size Min BA SMTI} belongs to $\FPT$ when parameterised by the number $k$ of different types in the given instance.
\end{theorem}
\begin{proof}
For any matching $M$, we can define a collection of at most $k(k+1)/2$ boolean variables $v_{i,j}$ (for $1 \leq i < j \leq k+1$, where type $k+1$ is a dummy type), so that $v_{i,j}$ is true if and only if the matching contains at least one pair involving an agent of type $i$ and an agent of type $j$ (and unmatched agents are considered to be matched with agents of the dummy type $k+1$).  For a given matching $M$, this collection of variables defines a vector $\mathbf{v}_M$ in $\{0,1\}^{k(k+1)/2}$, which we call the \emph{type-signature} of the matching $M$.

Note that there are at most $2^{k(k+1)/2} = \mathcal{O}(2^{k^2})$ possible type-signatures for a matching; we will consider each possible type-signature $\mathbb{v}$ in turn and determine the minimum number of agents which can be involved in blocking pairs in a maximum matching which has type-signature $\mathbb{v}$ (if a maximum matching with this type-signature exists).  Minimising over this set of optimal solutions will give the desired answer.

We now describe how to compute the minimum number of agents involved in blocking pairs in a maximum matching with type-signature $\mathbf{v}$ or else to report that no such maximum matching exists.  Our strategy is to encode the problem as an instance of \textsc{Integer Linear Programming}.

First we define the constraints.  As usual, we need to ensure that every  agent is involved in exactly one pair (potentially involving a dummy agent), and as in the proof of Theorem \ref{thm:min-bp-fpt} we need to enforce that the number of pairs that do not involve dummy agents is equal to the maximum cardinality of any matching in our instance.  Moreover, we need to make sure that our matching does indeed have type-signature equal to $\mathbf{v}$.  This gives rise to the following linear constraints.

\begin{align*}
~ & \sum_{1 \leq i < j \leq k} n_{\{i,j\}} = C_{\max}\\
~ & \sum_{j \in [k+1]} n_{\{i,j\}} = |N_i| & \forall i \in [k]\\
~ & n_{\{i,j\}} > 0 & \forall 1 \leq i < j \leq k+1 \text{ with } v_{i,j} = 1 \\
\text{\textbf{and}} ~ ~ & n_{\{i,j\}} = 0 & \forall 1 \leq i < j \leq k+1 \text{ with } v_{i,j} = 0.
\end{align*}
Finally, we define our objective function, which captures the number of agents which are involved in at least one blocking pair.  By Lemma \ref{lma:num-ba}, we know that an agent of type $i$ matched with an agent of type $j$ belongs to a blocking pair if and only if there exist $i \pref_{j'} i'$ and $j' \pref_i j$ 
such that $v_{i',j'} = 1$.  Thus, for a given type-signature $\mathbf{v}$, we can compute for each $1 \leq i,j \leq k+1$ the indicator variable $b_{i,j}$ which takes the value $1$ if an agent of type $i$ matched with an agent of type $j$ in a matching with type-signature $\mathbf{v}$ will belong to a blocking pair, and takes the value $0$ otherwise. It is now clear that the total number of agents that are involved in at least one blocking pair in the matching is
$$\sum_{1 \leq i < j \leq k+1} n_{\{i,j\}} (b_{i,j} + b_{j,i}).$$
This is our linear objective function.
\end{proof}

To generalise to the non-bipartite case takes just slightly more care: if two agents of type $i$ which are both matched to agents of type $j$ form a blocking pair, then the total number of blocking pairs that results is $n_{\{i,j\}}(n_{\{i,j\}}-1)/2$ rather than $n_{\{i,j\}}^2$.  Otherwise, exactly the same method works.  Thus we obtain the following corollary.

\begin{corollary}
\textsc{Typed Max Size Min BP SRTI} and \textsc{Typed Max Size Min BA SRTI} belong to $\FPT$ when parameterised by the total number of types.
\end{corollary}

It is also fairly straightforward to modify the IQP and ILP in the proofs of Theorem \ref{thm:min-bp-fpt} and \ref{thm:min-ba-fpt} to solve MIN BP SRTI and MIN BA SRTI respectively. We only need to remove the constraint that enforces the matching to be of size $C_{\max}$. 

\begin{corollary}
\textsc{Typed Min BP SRTI} and \textsc{Typed Min BA SRTI} belong to $\FPT$ when parameterised by the total number of types.
\end{corollary}

A related problem to MIN BP SR is EXACT BP SR which given an instance $I$ of SR and an integer $Z$, decides whether $I$ admits a matching with exactly $Z$ blocking pairs. Even this problem is NP-hard \cite{ABM06} in general. If $I$ is a typed instance, to solve EXACT BP SRTI we only have to move the objective function in the IQP of Theorem \ref{thm:min-bp-fpt} to the set of constraints, and enforce it to be equal to $Z$. 

\begin{corollary}
\textsc{Typed EXACT BP SRTI} belongs to $\FPT$ when parameterised by the total number of types.
\end{corollary}

\section{Agents of the same type refine their preferences in the same way}\label{sec:refined-typed}
In this section, we generalise the model from Section~\ref{sec:typed} by allowing agents to refine their preferences over candidates within a particular type, so long as agents of the same type still have identical preference lists. Our key result is that refining preferences in this way can never change the size of the largest stable matching, compared with the corresponding typed instance. We also use the tools we develop to deal with this generalisation to show that \textsc{Max SRTI}, and hence also \textsc{Max SMTI} and \textsc{Max HRT}, become polynomially solvable if preferences over types are strict, both in this setting and under the basic model. 
Lastly, we extend the results in Section \ref{sec:typed-minbp} to provide FPT algorithms for \textsc{Consistently-refined-typed Max Size Min BP SMTI} and \textsc{Consistently-refined-typed Max Size Min BA SMTI}.

\subsection{Definition of consistently-refined-typed instances}\label{sec:def-refined-model}
Consider a generalisation of typed instances in which agents are no longer necessarily indifferent between two agents of the same type, however agents of the same type occur consecutively in preference lists. 
This means that for any two agents $x$ and $y$ of the same type $i$:
\begin{enumerate}
  \item $x$ and $y$ have identical preference lists when restricted to $\agentset \setminus \{x,y\}$
  ,
  \item no agent of a different type appears between $x$ and $y$ in any preference list, and
  \item if a tie in a preference list contains agents of two or more types, then that tie is in fact a union of types.
\end{enumerate}

The third criterion allows us to define in a consistent way what it means for agents of type $i$ to strictly prefer type $j$ to type $\ell$ or to be indifferent between them. We will say that agents of type $i$ prefer type $j$ to type $\ell$ if 
and only if given every pair of agents $x$ of type $j$ and $y$ of type $\ell$ all agents in $N_i$ prefer $x$ to $y$. 
On the other hand, if type $i$ is indifferent between types $j$ and $\ell$ it means that, in the preference list for each agent $x$ of type $i$, all agents in $N_j\cup N_{\ell}$ belong to a single tie.

If an instance of a stable matching problem satisfies these slightly weaker requirements, we say that the instance is \emph{consistently-refined-typed}, and refer to the standard problems with input of this form as \textsc{Consistently-Refined-Typed Max SMTI} etc.

A consistently-refined-typed instance $I$ of SRTI is given as an input by specifying the number of types $k$ and, for each type $i$, the set $\agentset_i$ of agents of type $i$ as well as the preference ordering $\succ_i$ over agents. Note that for typed instances $\succ_i$ specified preferences over types, whereas here the preferences are over agents. However, we can compute preferences over types from preferences over agents in time $\mathcal{O}(kn)$. Note that if we are only given the preference list for each agent as input (i.e., no information about types is given), it is straightforward to compute, in polynomial time, the coarsest partition of the agents into types that satisfies the definition of consistently-refined-typed instance.

We illustrate this definition with two short examples.

\begin{example}\label{ex:model2-1}
Assume we have 4 types for the agents in a stable marriage setting, and that all men are of type $1$ and types $2$, $3$, and $4$ correspond to women. Assume also that we have 3 men $m_1$, $m_2$ and $m_3$, and 7 women where $w_1$ and $w_2$ are of type $2$, $w_3$ and $w_4$ are of type $3$, and $w_5$, $w_6$ and $w_7$ are of type $4$. Let all men have the preference ordering $(w_1 ~ w_2 ~ w_3 ~ w_4) ~ w_6 ~ (w_5 ~ w_7) $, women of types $2$ and $3$ have the preference ordering $(m_1 ~ m_2 ~ m_3)$, and women of type $4$ have the preference ordering $m_2 ~ m_1$. This setting constitutes a consistently-refined-typed instance. It is easy to compute the preferences of type $1$ agents over the types of women, which is $(2 ~ 3) ~ 4$, similar to that of Example~\ref{ex:model1}. Allowing men to have the preference ordering $(w_1 ~ w_2) ~ w_3 ~ w_4 ~ (w_5 ~ w_6 ~ w_7)$, while keeping everything else unchanged, also gives us a consistently-refined-typed instance. In this new instance agents of type $1$ have the strict preference ordering $2 ~ 3 ~ 4$ over the types of women.
\end{example}

\begin{example}\label{ex:model2-2}
Assume that we are in a stable roommates setting with six agents $a, b, \ldots, f$ of type 1. assume that the agents' preference orderings are as follows:
\begin{align*}
		a:& ~ (b ~ c) ~ (d ~ e ~ f)  & b:& ~ (a ~ c) ~ (d ~ e ~ f) & c:& ~ (a ~ b) ~ (d ~ e ~ f) \\ 
    d:& ~ (a ~ b ~ c) ~ (e ~ f) & e:& ~ (a ~ b ~ c) ~ (d ~ f) & f:& ~ (a ~ b ~ c) ~ (d ~ e) 
\end{align*}
This setting constitutes a consistently-refined-typed instance and it is easy to compute the refined preferences within type 1 which is $(a ~ b ~ c) ~ (d ~ e ~ f)$. 
\end{example}

\subsection{An FPT algorithm for \textsc{Consistently-Refined-Typed Max SRTI}}\label{sec:con-refined-maxsrti}
To extend the result for \textsc{Typed Max SRTI} to \textsc{Consistently-Refined-Typed Max SRTI}, we need the following result.

\begin{lemma}\label{stable-exists}
Let $I$ be a consistently-refined-typed instance of SRTI and suppose that $M$ is a matching in $I$ such that (1) there is no pair $(i,j) \in [k]^{(2)}$, $i\neq j$, where $j \pref_i \worst_M(i)$ and $i \pref_j \worst_M(j)$, and (2) there is no pair $(i,i)$, $i \in [k]$, such that there is at least two agents of type $i$ and $i \pref_i \secondworst(i)$.
Then there is a stable matching $M'$ such that, for every $(i,j) \in [k]^{(2)}$, both $M$ and $M'$ contain the same number of pairs that consist of one agent of type $i$ and another of type $j$.  Moreover, given $M$, we can compute $M'$ in time $\mathcal{O}(kn)$.
\end{lemma}

\begin{proof}
Let $n_{\{i,j\}}(M)$ denote the number of pairs in $M$ consisting of an agent of type $i$ and an agent of type $j$. We construct a stable matching $M'$ such that $n_{\{i,j\}}(M')=n_{\{i,j\}}(M)$. 
Let $\bigcup_{i,j\in [k], i\leq j} M'_{i,j} = M'$ be a decomposition of $M'$ where $M'_{i,j}$ is the projection of $M'$ onto agent types $i$ and $j$. The polynomial-time construction of $M'$ takes place in two steps. 

\vspace{10pt}

\textbf{Step 1:} To start with, all agents are available. For each type $i$, take the candidate types in type $i$'s decreasing order of preference with ties broken arbitrarily, $\langle i_1, i_2, \ldots i_{k_i}\rangle$, where $i_s$ denotes the type that is ranked $s$'th by type $i$. Starting with $j=i_1$, if $j\neq i$, take the topmost (from the perspective of an agent of type $j$) $n_{\{i,j\}}(M)$ available agents of type $i$, and put them in $A_{i,j}$; these agents become unavailable from now on. If $j=i$ then take the topmost $2 \cdot n_{\{i,i\}}(M)$ available agents of type $i$, and put them in $A_{i,i}$. $A_{i,j}$ includes agents of type $i$ that are to be matched to agents of type $j$. Note that as the preference lists of agents of type $j$ over agents of type $i$ may include ties, it may not be possible to determine exactly who are the topmost available $n_{\{i,j\}}(M)$ agents in type $i$. To be more precise, when going down the preference list of agents of type $j$ over available agents of type $i$, we may reach a tie $\tau$ including $z>1$ available agent where we need to pick $x<z$ number of them. If this happens, arbitrarily pick $x$ agents from $\tau$. This step can be done in time $\mathcal{O}(kn)$. 

\vspace{10pt}

\textbf{Step 2:} We now show how to generate $M'_{i,j}$ given $A_{i,j}$ and $A_{j,i}$ computed in Step 1. We do so by computing a complete stable matching amongst the agents in $A_{i,j} \cup A_{j,i}$.  To do this, we list the agents in $A_{i,j}$ in non-increasing order of preference with respect to $j$, and similarly for $A_{j,i}$, then pair up each agent with the agent having the same position on the other list.  The total time required to do this for all pairs $(i,j)$ is $\mathcal{O}(n)$.

\vspace{10pt}
It follows immediately from this construction that, for each $i$ and $j$, $n_{\{i,j\}}(M')=n_{\{i,j\}}(M)$.
So it only remains to prove that $M'$ is stable. Assume for a contradiction that $M'$ admits a blocking pair $(a,b)$ where $a$ , $b$, $M'(a)$ and $M'(b)$ are of types $i$, $j$, $j'$ and $i'$ respectively. Five ``kinds'' of blocking pairs are possible, depending on how $a$ and $b$ compare each others' types against the types of their partners. We examine each of them and show that $M'$ can admit none. 

\begin{itemize}
	\item \emph{Suppose that $i\neq j$, $a$ prefers type $j$ to type $j'$, and $b$ prefers type $i$ to type $i'$.} In the assumption of the lemma we have that in the given $M$ there is no pair $(i,j) \in [k]^{(2)}$ where $j \pref_i \worst_M(i)$ and $i \pref_j \worst_M(j)$. By construction of $M'$, $\worst_M(i)$ remains unchanged under $M'$ for all types $i$. It thus directly follows that $M'$ cannot admit such a blocking pair.
	 \item \emph{Suppose that $j\tie_{i} j'$ ($j\neq j'$) or $i\tie_{j} i'$ ($i\neq i'$).}  We may assume that $j\tie_{i} j'$, $j\neq j'$; a symmetric argument holds when $i\tie_{j} i'$, $i\neq i'$.  Then, by our assumption that the instance is consistently-refined-typed, any agent of type $i$, and hence $a$, is indifferent between all agents who are of type $j$ or type $j'$. Therefore, $M(a) \tie_{a} b$ and $(a,b)$ cannot be a blocking pair. 
  	\item \emph{Suppose that $i=i'$ and $j=j'$.} The existence of such a blocking pair implies that $M'_{i,j}$ constructed in Step 2 is not  stable with respect to the preferences of agents in $A_{i,j} \cup A_{j,i}$, a contradiction.
    \item \emph{Suppose that either $i=i'$ or $j=j'$, but not both.} Without loss of generality assume that $j=j'$ and $b$ prefers type $i$ to type $i'$. Since $a$ prefers $b$ to $M'(a)$, and $M'(a) \in A_{j,i}$, it follows from the construction in Step 1 that, since $b$ is not in $A_{j,i}$ and is in $A_{j,i'}$, any agent of type $j$ (including $b$) either prefers type $i'$ to type $i$ or is indifferent between them, a contradiction. 
		\item \emph {Suppose that $i=j$ so that $a$ and $b$ are of the same type.} In the assumption of the lemma we have that given $M$ there is no pair $(i,i) \in [k]^{(2)}$ where $i \pref_i \secondworst_M(i)$. By the construction of $M'$, $\secondworst_M(i)$ remains unchanged under $M'$ for all types $i$. It thus directly follows that $M'$ cannot admit such a blocking pair.
\hfill \qedsymbol
\end{itemize}
\renewcommand{\qedsymbol}{}
\end{proof}

Let $I$ be a consistently-refined-typed instance of SRTI and let $I'$ be a typed instance of SRTI that is obtained from $I$ by ignoring the refined preferences within each type (i.e. every agent is indifferent between the candidates of the same type). It follows from the definition of stability that every matching that is stable in $I$ is also stable in $I'$. Lemma \ref{stable-exists} implies that for any stable matching $M$ in $I'$, there exists a stable matching $M'$ in $I$ of the same cardinality as $M$. Thus, in order to find a maximum cardinality matching in a consistently-refined-typed instance $I$ of SRTI, it suffices to (1) solve the typed problem (i.e. ignore the refined preferences within each type) and then (2) use the algorithm provided in the proof of Lemma \ref{stable-exists} to convert the solution to a matching of the same cardinality that is stable in the instance $I$. 
Deriving a typed instance from a consistently-refined-typed instance can be done easily in time $\mathcal{O}(kn)$. It thus follows that \textsc{Consistently-Refined-Typed Max SRTI} is in \cplxty{FPT} parameterised by the number $k$ of different types in the instance.
\begin{theorem}\label{model2-maxsrti}
\textsc{Consistently-Refined-Typed Max SRTI} can be solved in time $k^{\mathcal{O}(k^2)}\log^3 n + \mathcal{O}(kn)$.
\end{theorem}

We also have the following immediate corollary.

\begin{corollary}\label{model1-2-maxhrt}
\textsc{Consistently-Refined-Typed Max SMTI} and \textsc{Consistently-Refined-Typed Max HRT} are in $\FPT$ parameterised by the number $k$ of different types in the instance.
\end{corollary}

\subsection{Strict preferences over types}
Elsewhere in the paper, we assume that agents can be indifferent between agents of two or more types.  We now show that \textsc{Max SRTI} becomes easier if we restrict the set of possible instances by assuming that agents have strict preferences over types. 
This argument is based on a private communication with David Manlove.
\begin{theorem}\label{thm:refined-smti-strict}
When preferences over types are strict, \textsc{Typed Max SRTI} and \textsc{Consistently-Refined-Typed Max SRTI} are polynomial-time solvable. Furthermore, all stable matchings (if any exists) are of the same size.
\end{theorem}

\begin{proof}
Since \textsc{Typed Max SRTI} is a special case of \textsc{Consistently-Refined-Typed Max SRTI}, it suffices to demonstrate that the result holds for the latter problem.
Let $I$ be an instance of \textsc{Consistently-Refined-Typed Max SRTI}, and let $I'$ be another instance of \textsc{Consistently-Refined-Typed Max SRTI} obtained from $I$ by breaking any remaining ties in the preference lists arbitrarily and consistently (so that agents of the same type still have identical preference lists).  Note that $I'$ is in fact an instance of SRI. 
All stable matchings in an instance of SRI (if it admits any) have the same cardinality \cite{fGI89} and so we can find a maximum cardinality stable matching in $I'$, or report that none exists, in polynomial time \cite{fGI89}. We will argue that if $I'$ admits no stable matching then neither does $I$,  otherwise, a maximum cardinality stable matching $M$ in $I'$ is in fact a maximum cardinality stable matching in $I$. Therefore \textsc{Consistently-Refined-Typed Max SRTI} is polynomial-time solvable.

Consider $I_0$, the typed instance of SRTI obtained from $I$ by ignoring preferences within types (i.e. placing all agents of the same type in a single tie).  Note that (as preferences over types are strict in $I$) if we apply the same process to $I'$, we also obtain the same typed instance $I_0$.
By Lemma \ref{stable-exists} we know that, given any stable matching $M_0$ in $I_0$, there is a stable matching $M'$ in $I'$ of the same cardinality as $M_0$. Therefore, if $I'$ admits no stable matching neither does $I_0$. Furthermore, if $I_0$ admits no stable matching then it follows from the definition of stability that neither does $I$.

Now assume that $I'$ admits a stable matching. It is clear that $M$ must be stable in $I$, since (as each agent's preference list in $I'$ is obtained from that in $I$ by breaking ties) any blocking pair with respect to $M$ in the instance $I$ would also be a blocking pair in $I'$. It remains to argue that there cannot be any stable matching of cardinality greater than $|M|$ in $I$. 
Since all stable matchings in $I'$ have the same cardinality, we can conclude that (following Lemma \ref{stable-exists}) every stable matching in $I_0$ has the same cardinality as $M$.  Moreover, as any matching that is stable in $I$ must also be stable in $I_0$, it follows that all stable matchings in $I$ have the same cardinality as $M$. Thus $M$ is a maximum cardinality stable matching in $I$, as required.
\end{proof}

The above result, combined with Corollary \ref{model1-2-maxhrt}, gives us the following result.

\begin{corollary}\label{thm:hrt-strict}
When preferences over types are strict, \textsc{(Consistently-Refined-)Typed Max SMTI} and \textsc{(Consistently-Refined)-Typed Max HRT} are polynomial-time solvable. Furthermore, all stable matchings are of the same size.
\end{corollary}

\subsection{FPT algorithms for \textsc{Consistently-refined-typed Max Size Min BP SMTI} and \textsc{Consistently-refined-typed Max Size Min BA SMTI}}\label{sec:refined-typed-minbp}

In this section we show how we can apply the ideas from the previous section to solve \textsc{Max Size Min BP/BA SMTI} for consistently refined typed instances.

\begin{corollary}
\textsc{Consistently-refined-Typed Max Size Min BP SMTI} and \textsc{consistently-refined-Typed Max Size Min BA SMTI} belong to $\FPT$ when parameterised by the number of types.
\end{corollary}
\begin{proof}
We use the same strategy for both problems. First, we solve the problem for the corresponding typed instance (ignoring the more refined preferences within types).  Note that the number of blocking pairs or blocking agents achieved in this simplification of the problem clearly gives a lower bound on the minimum number that can be achieved if we take into account all information; we will argue that in fact we can always obtain a matching which does not increase either quantity when we take into account the full preference lists.

To do this, we follow the method described in Lemma \ref{stable-exists}.  Given the number of pairs $n_{\{i,j\}}$ of type $i$ and $j$ for each $1 \leq i < j \leq k$, this method allows us to construct a matching $M$ where, for each $1 \leq i < j \leq k$ we have exactly $n_{\{i,j\}}$ pairs involving one agent of type $i$ and one of type $j$, and there is no blocking pair $(x,y)$ such that $x$ is currently matched to an agent of the same type as $y$.  Thus the only blocking pairs in $M$ are those of the form $(x,y)$ where $x$ is of type $i$ and $y$ of type $j$, and $j \pref_i \type(M(x))$ and $i \pref_j \type(M(y))$.  But these are precisely the blocking pairs that occur in the relaxation to a typed instance.

Thus we can indeed obtain a solution to \textsc{Max Size Min BP SMTI} or \textsc{Max Size Min BA SMTI} by applying the appropriate algorithm to find the number of pairs of each type under the relaxation to a typed instance, and then use the method of Lemma \ref{stable-exists} to extend this to a matching which does not introduce any additional blocking pairs when the full preference lists are taken into consideration.
\end{proof}

Combining this result with the techniques of Section \ref{sec:typed-minbp}, we obtain the following corollary.

\begin{corollary}
The following problems all belong to $\FPT$ when parameterised by the total number of types:
\begin{itemize}
\item \textsc{Consistently-refined-Typed Max Size Min BP SRTI},
\item \textsc{Consistently-refined-Typed Max Size Min BA SRTI},
\item \textsc{Consistently-refined-Typed Min BP SRTI},
\item \textsc{Consistently-refined-Typed Min BA SRTI}, and
\item \textsc{Consistently-refined-Typed EXACT BP SRTI}.
\end{itemize}
\end{corollary}

\section{Exceptions in preference lists}
\label{sec:exceptions}

We have argued for the existence of typed instances, where $k \ll n$, based on the premise that agents' preferences are formed based on a small collection of candidates' attributes. 
In practice, it seems likely that an agent might have access to additional information about some small subset of the candidates, either through personal acquaintance or some third-party connection; we say that an agent considers such candidates to be \emph{exceptional}. This additional information may alter the agent's opinion of candidates relative to that derived from the attributes alone, and so affect where these candidates are placed in his/her preference ordering.

In this section we consider a generalisation of typed instances in which each agent may find some small collection of other agents to be exceptional and ranks them without regard to their types. Note that if only a small number of the agents in our instance consider one or more candidates to be exceptional, we can capture this information in a typed instance: each agent with exceptions in their preference list can be assigned their own type.

We say that an instance $I$ of a stable matching problem is a \emph{(c,Any)-exception-typed} instance, for a given constant $c$, if $I$ is a typed instance in which each agent finds at most $c$ number of the candidates exceptional and may rank them anywhere in his/her preference list. 
Two special cases are \emph{(c,Top)-exception-typed} and \emph{(c,Bottom)-exception-typed} instances where the exceptions are promoted to the top, or demoted to the bottom, of the preference lists, respectively. We refer to the standard problems with input of this form as \textsc{(c,Any)-Exception Typed Max SMTI} etc.

We show that if each agent finds at most one candidate exceptional, whom s/he promotes to the top of his/her list, then \textsc{Max SMTI} belongs to \cplxty{FPT}. In contrast, if we allow for two (or more) exceptions, and exceptional candidates can appear anywhere in the preference lists, then \textsc{Max SMTI} remains \cplxty{NP}-hard even when there are only a constant number of types.

\subsection{An FPT algorithm for \textsc{(1,Top)-Exception Typed Max SMTI}}\label{sec:1-top-exceptions}

We consider a generalisation of typed instances in which (i) every agent considers at most one other agent to be exceptional, and (ii) every agent would strictly prefer to be assigned to the agent they consider to be exceptional than to any other agent.
For each agent $a$ let $\ext(a)$ denote the exceptional candidate from $a$'s point of view; $\ext(a)=\varnothing$ if $a$ does not find any candidate exceptional. Then, $I$ is a \emph{(1,Top)-exception-typed} instance of SMTI if, given every two agents $x$ and $y$ of the same type:
\begin{enumerate}
\item $x$ and $y$ have identical preference lists when restricted to $\agentset \setminus \{\ext(x),\ext(y)\}$, and
\item all other agents who do not find either $x$ or $y$ exceptional are indifferent between $x$ and $y$.
\end{enumerate}

Without loss of generality we can assume that there is no pair of agents who each consider the other to be exceptional in a (1,Top)-exception-typed instance of SMTI.  If there are such pairs, they must be assigned to each other in any stable matching; so we can remove all such pairs to reduce to an instance that satisfies this assumption.

A (1,Top)-exception-typed instance of SMTI is given as input by, in addition to the specifications needed for a typed instance (see Section~\ref{sec:def-basic-model}), providing for each agent his or her exceptional candidate (if s/he has one).

Let $I$ be a (1,Top)-exception-typed instance of SMTI, and let $M$ be a matching in $I$.  As in Section \ref{sec:typed-maxsrti}, we may assume without loss of generality that every agent is matched, by creating sufficiently many dummy agents of type $k+1$ which are inserted at the end of each man's and woman's (possibly incomplete) preference list. In order to obtain an analogue of the stability criterion given in Lemma \ref{stability-test-srti} in this setting, we need some more notation.  

Recall that we write $j \tie_i \ell$ if agents of type $i$ are indifferent between types $j$ and $\ell$. It is straightforward to see that $\tie_i$ defines an equivalence relation on $[k]$ for each $i$.  Given $j \in [k]$, we write $\cls_i(j)$ for the equivalence class under $\tie_i$ which contains $j$.  For each equivalence class $J$ under $\tie_i$, we say that the agent $x$ of type $i$ has \emph{subtype} $i[J]$ if:
\begin{enumerate}
\item some agent $y$, with $\type(y) \in J$, considers $x$ exceptional, and
\item there is no agent $z$, such that $\type(z) \pref_i j$ for $j \in J$, who considers $x$ exceptional.
\end{enumerate}
Thus $\subtype(x) = i[J]$ if the most desirable agents who consider $x$ exceptional have types from $J$.  If an agent $x$ of type $i$ is not considered exceptional by any agent, we say that $x$ has subtype $i[\{k+1\}]$. We also introduce a second dummy type $0$, which is inserted at the head of each type's preference list and corresponds to exceptional candidates. We write $N_{i[J]}$ for the set of agents of subtype $i[J]$.

Observe that the sets $N_{i[J]}$ can be computed in time $\mathcal{O}(n)$: for each agent $x$, $\subtype(x)$ can be computed in time $\mathcal{O}(n)$ with suitable data structures.

We will need a variation on the function $\worst_M$, which we call $\exworst_M$.  For any non-empty set $N_{i[J]}$, $\exworst_M(i[J])$ is the type of the least desirable partner received by an agent of subtype $i[J]$ who is not matched with an agent they find exceptional; if the least desirable partners assigned to agents of subtype $i[J]$ belong to two or more different types between which agents of type $i$ are indifferent, we define $\exworst_M(i[J])$ to be the lexicographically first such type.  If every agent of subtype $i[J]$ is matched with a partner they find exceptional, we set $\exworst_M(i[J]) = 0$.  Therefore $\worst_M(i)$, as defined in Section \ref{sec:typed-maxsmti}, is the least desirable type out of $\{\exworst_M(i[J]): N_{i[J]} \neq \emptyset\}$.

We say that a matching $M$ in an instance $I$ of (1,Top)-exception-typed SMTI realises the function $\exworst$, mapping nonempty subtypes $i[J]$ to values in $\{0,1,\ldots,k+1\}$, if $\exworst_M(i[J]) \weaklypref_i \exworst(i[J])$ whenever $N_{i[J]} \neq \emptyset$.  We can now characterise stability in a (1,Top)-exception-typed instance.

\begin{lemma}\label{lma:exception-stab-test}
Let $I$ be a (1,Top)-exception-typed instance of SMTI.  Then a matching $M$ in $I$ is stable if and only if there is no pair $(i,j) \in [k]^{(2)}$ such that
\begin{enumerate}
\item  $j \pref_i \worst_M(i)$ and $i \pref_j \worst_M(j)$, or
\item $i \pref_j \exworst_M(j[\cls_j(i)])$.
\end{enumerate}
\end{lemma}
\begin{proof}
Suppose first that $M$ is not stable. In this case there is a blocking pair $(x,y)$ where $x$ is of type $i$ and $y$ is of type $j$. We show that at least one of the two aforementioned conditions must hold. If neither $x$ nor $y$ finds the other exceptional, then it must be that they prefer each other's type to the type of their assigned partner. This implies that $j \pref_i \type(M(x)) \weaklypref_i \worst_M(i)$ and $i \pref_j \type(M(y)) \weaklypref_j \worst_M(j)$ so condition (1) holds. If one of $x$ and $y$ finds the other exceptional, we may assume without loss of generality that $x$ finds $y$ exceptional. We may further assume, without loss of generality, that $x$ is the (equally) most desirable agent (from the point of view of $y$) who prefers $y$ to his/her current partner.  
This means there is no agent of type $\ell$, with $\ell \pref_j \type(x) = i$, who considers $y$ to be exceptional, 
implying that $y \in N_{j[\cls_j(i)]}$.
Hence, by the definition of a blocking pair, we see that $i = \type(x) \pref_j \type(M(y)) \weaklypref_j \exworst_M(j[\cls_j(i)])$, and condition (2) holds.

Conversely, it is straightforward to see that if we have any pair $(i,j)$ satisfying condition (1) or (2) then there will be a blocking pair and hence $M$ will not be stable. If condition (1) holds, then $\worst_M(i)$ and $\worst_M(j)$ must both be greater than $0$, and the condition implies that there is a pair $(x,y)$ of type $(i,j)$ who prefer each other to their partners.  If condition (2) holds, then there exists an agent $y \in N_{j[\cls_j(i)]}$ who is considered exceptional by some agent $x$ 
whose type is in $\cls_j(i)$, and $y$ prefers any agent whose type belongs to $\cls_j(i)$
to his/her current partner, so $(x,y)$ forms a blocking pair.
\end{proof}

We will say that a function $\exworst$ is \emph{$I$-exception-stable} for a (1,Top)-exception-typed instance $I$ of SMTI if there is no pair $(i,j) \in [k]^{(2)}$ such that either $j \pref_i \worst(i)$ and $i \pref_j \worst(j)$ or $i \pref_j \exworst(j[\cls_j(i)])$. 
Observe that Lemma \ref{lma:exception-stab-test} gives us the following straightforward corollary.

\begin{corollary}\label{cor:best-exception}
Let $I$ be a (1,Top)-exception-typed instance of SMTI.  Suppose that $M$ is a matching in $I$ which realises the $I$-exception-stable function $\exworst$, and that $M$ contains a pair $(x,y)$ where $x$ considers $y$ to be exceptional and $y$ is of subtype $i[J]$.  Then $\type(x) \in J$ and $\exworst(i[J]) \in J$.
\end{corollary}
\begin{proof}
Since $\subtype(y) = i[J]$, there exists $j \in J$ such that some agent $z$ of type $j$ considers $y$ to be exceptional.  By definition of subtype, the fact that $x$ considers $y$ to be exceptional means that $j \succeq_i \type(x)$.  Thus, as $M$ realises $\exworst$, we see that 
$$j \weaklypref_i \type(x) = \type(M(y)) \weaklypref_i \exworst(i[J]).$$
If any of the preferences in this equation is in fact strict, we would violate condition (2) of Lemma \ref{lma:exception-stab-test}, contradicting our assumption that $\exworst$ is $I$-exception stable.  Therefore we must have $j  \tie_i \type(x) \tie_i \exworst(i[J])$, implying that 
both $\type(x)$ and  $\exworst(i[J])$ are in $\cls_i(j) = J$, as required.
\end{proof}

Given any $I$-exception-stable function $\exworst$, we write $\max(\exworst)$ for the maximum cardinality of any matching in $I$ that realises $\exworst$.  We have an analogous result to Corollary \ref{lma:max-I-stable} in this setting.

\begin{lemma}\label{lma:max-I-exstable}
Let $I$ be a (1,Top)-exception-typed instance of SMTI.  Then the cardinality of the largest stable matching in $I$ is equal to 
$$\max\{\max(\exworst): \exworst \text{ is $I$-exception-stable}\}.$$
\end{lemma}

We now argue that we can compute $\max(\exworst)$ for any $I$-exception-stable function $\exworst$ in polynomial time.

\begin{lemma}\label{lma:dec-realising-exworst}
Let $I$ be a (1,Top)-exception-typed instance of SMTI, and fix an $I$-exception-stable function $\exworst$.
We can compute  $\max(\exworst)$, and generate a stable matching of size $\max(\exworst)$, in time $\mathcal{O}(n^{5/2}\log n)$.
\end{lemma}
\begin{proof}[Proof sketch.] 
Using a binary search strategy, we can determine the maximum size of a matching realising $\exworst$ by solving $\mathcal{O}(\log n)$ instances of the decision problem ``Is $\max(\exworst)$ at least $c$?", where $c \in \{1,\ldots,\lfloor n/2 \rfloor\}$. We show that to determine whether $\max(\exworst) \geq c$ it is enough to decide whether an undirected graph $G=(V,E)$ with $\mathcal{O}(n)$ vertices, constructed from $I$ in time $\mathcal{O}(n^2)$, admits a perfect matching. Moreover, if $G$ admits a perfect matching $E'$, $E'$ corresponds to a stable matching in $I$ of size at least $c$. The graph $G=(V,E)$ is constructed as follows. Suppose that there are in total $n$ agents. The vertex-set $V$ contains one vertex corresponding to each agent in $I$, together with $n-2c$ dummy vertices.  The edge-set $E$ contains an edge between the agents $x$, of type $a[C]$, and $y$, of type $b[D]$, if and only if $x$ and $y$ find each other mutually acceptable and either
\begin{enumerate}
\item neither $x$ nor $y$ considers the other to be exceptional, $b \weaklypref_a \exworst(a[C])$, and $a \weaklypref_b \exworst(b[D])$, or
\item $a \in D$, $x$ considers $y$ to be exceptional, and $\exworst(b[D]) \in D$, or
\item $b \in C$, $y$ considers $x$ to be exceptional, and $\exworst(a[C]) \in C$.
\end{enumerate}
We additionally have an edge between each pair of dummy vertices, and between each dummy vertex and every agent $x$ such that 
$\exworst(i[J]) = k+1$ where $i[J]=\subtype(x)$; that is, if $\exworst$ allows for $x$ to be unmatched.
Deciding whether $G$ admits a perfect matching can be done in time 
$\mathcal{O}(\sqrt{|V|}|E|)$, which coupled with the fact that $|V|=\mathcal{O}(n)$ gives us 
$\mathcal{O}(n^{5/2})$. Therefore, we conclude that we can compute $\max(\exworst)$, and generate a stable matching of size $\max(\exworst)$, in time $\mathcal{O}(n^{5/2}\log n)$. 
\end{proof}

\begin{corollary}
\textsc{(1,Top)-Exception Typed Max SMTI} can be solved in time $\mathcal{O}\left(k^{k^2}(k^3 + n^{5/2}\log n)\right)$.
\end{corollary}
\begin{proof} Suppose that $I$ is the input to our instance of \textsc{(1,Top)-Exception Typed Max SMTI}. We consider each possible function $\exworst: [k]\times [k] \rightarrow \{0,1,\ldots,k+1\}$ in turn; there are at most $(k+2)^{k^2}$ such functions.  We can determine in time $\mathcal{O}(k^3)$ whether $\exworst$ is $I$-exception-stable. For each $I$-exception-stable function $\exworst$, we compute $\max(\exworst)$ in time $\mathcal{O}(n^{5/2}\log n)$ by Lemma ~\ref{lma:dec-realising-exworst}. We then take the maximum value of $\max(\exworst)$ over all $I$-exception-stable functions $\exworst$ which, by Lemma \ref{lma:max-I-exstable}, is equal to the cardinality of the largest stable matching in $I$. The proof of Lemma ~\ref{lma:dec-realising-exworst} is constructive, that is, it also generates a stable matching of size $\max(\exworst)$.
\end{proof}

\subsection{\textsc{(2,Any)-Exception Typed Max SMTI} is \cplxty{NP}-hard for a constant number of types}
\label{sec:2-exception-hard}

In this section we investigate the effects of allowing each agent to consider more than one candidate to be exceptional. 
We show that, if we allow each agent to declare two candidates exceptional, and these two candidates can appear anywhere in the agent's preference list, then \textsc{Max SMTI} is NP-hard even under severe restrictions. In fact, we show that the special case \textsc{Com SMTI}, which involves deciding whether a given instance of SMTI admits a complete stable matching (i.e., a matching that matches all agents), is \cplxty{NP}-complete even if the number of types is bounded by a constant. 
Our proof is by reduction from the \cplxty{NP}-complete problem \textsc{Clique}, defined as follows: given an undirected graph $G = (V,E)$ and $r \in \mathds{N}$, does $G$ contain a clique on at least $r$ vertices?

\begin{theorem}\label{2-exception-com-smti-hard}
\textsc{(2,Any)-Exception Typed COM SMTI} is \cplxty{NP}-complete, even if only men have exceptions in their preference lists, preferences over types are strict, and there are three types each of men and women.
\end{theorem}
\begin{proof}
The problem obviously is in \cplxty{NP} as we can easily check in quadratic time whether a given matching is complete and stable. To prove \cplxty{NP}-hardness, we reduce from \textsc{Clique}.
Let $(G=(V,E),r)$ be the input to an instance of \textsc{Clique}, and suppose that $V = \{v_1,\ldots,v_n\}$ and $E = \{e_1,\ldots,e_m\}$.  We will construct an instance $I$ of \textsc{(2-Any)-Exception Typed Max SMTI} so that $G$ contains a clique on $r$ vertices if and only if there is a complete stable matching in $I$. 

Our instance $I$ has six types in total; types 1, 2 and 3 are types of men and types 4, 5 and 6 are types of women.  Type 1 contains one man corresponding to each edge in $E$; abusing notation, we shall write $e_i$ for the man corresponding to the edge $e_i$.  Type 4 contains one woman corresponding to each vertex in $V$, and again we abuse notation to write $v_i$ for the woman corresponding to the vertex $v_i$.  Types 2 and 3 consist of $r$ and $n-r$ men respectively, and types 5 and 6 consist of $\binom{r}{2}$ and $m - \binom{r}{2}$ women respectively.\footnote{If there is a clique on $r$ vertices then there must be at least $\binom{r}{2}$ number of edges, hence $m -\binom{r}{2} \geq 0$. }  The preferences over types for each type is as follows:

\begin{center}
\begin{tabular}{l l l l l l}
Type 1: & 6 5 4 & & Type 4: & 2 1 3\\
Type 2: & 4 & &Type 5: & 1\\
Type 3: & 4 & & Type 6: & 1\\
\end{tabular}
\end{center}

We now describe the exceptions.  Only men of type 1 have any exceptions in their preference lists.  Each man $e_i$ of type 1 considers two women of type 4 to be exceptional, and he ranks these two women in a tie between types 6 and 5.
If $e_i = (v_j,v_{\ell})$ then $e_i$ considers $v_j$ and $v_{\ell}$ exceptional.
This completes the definition of $I$. Clearly $I$ can be constructed from $(G,r)$ in polynomial time.

We claim that $G$ contains a clique on $r$ vertices if and only if $I$ admits a complete stable matching.

$(\Leftarrow$):
We argue that if $(G,r)$ is a yes-instance, then $I$ admits a complete stable matching.  Let $U \subset V$ be a set of $r$ vertices in $G$ which induces a clique.  We construct matching $M$ in $I$ as follows. Each woman in $U$ (who are all of type 4) is matched to a man of type 2. Let $F$ denote the set of $\binom{r}{2}$ edges with both endpoints in $U$. Each man in $F$ (who is of type 1) is matched to a woman of type 5.  The remaining women of type 4 are matched to men of type 3, and the remaining men of type 1 are matched to women of type 6.  $M$ clearly matches all agents. Note that the matched pairs are all of the following four pairs of types: (1,5), (1,6), (2,4), and (3,4). It remains to show that $M$ is stable.

Note that no agent of type 2, 3, 5 or 6 can be involved in a blocking pair: all these agents are matched in $M$, and are indifferent between all candidates they find acceptable. Hence any blocking pair is between a man of type 1 and a woman of type 4. Moreover, no blocking pair can involve a man $m$ of type 1 who is matched with a woman of type 6, or a woman $w$ of type 4 who is matched with a man of type 2:  such agents ($m$ or $w$) are already matched with their first (equal) choice of partner.  Therefore, if there is a blocking pair, it is of the form $(e,v)$ where $e$ is a man of type 1, $v$ is a woman of type 4, $M(e)$ is of type 5 and $M(v)$ is of type 3.  By construction of $M$, as $M(e)$ is of type 5, we know that $e \in F$. Thus, the two women that $e$ finds exceptional both belong to $U$ and so are matched to men of type 2. It follows that $e$ does not find $v$ exceptional.  Thus $e$ does not prefer $v$ (of type 4) to his current partner of type 5, and so $(v,e)$ cannot form a blocking pair.  We can therefore conclude that $M$ is indeed stable.

$(\Rightarrow$):
Conversely, we show that if there is a complete stable matching $M$ in $I$, then $(G,r)$ is a yes-instance. Since $M$ matches every agent, 
all men of types 2 and 3 must be matched to women of type 4 (i.e. women corresponding to vertices), and all women of types 5 and 6 must be matched to men of type 1 (i.e. men corresponding to edge). The numbers of agents in each type mean that in a complete matching no man corresponding to an edge can be matched to a woman corresponding to a vertex.
Let $W$ be the set of $r$ women who are matched to men of type 2. We will argue that $W$ must induce a clique in $G$.  

Suppose, for a contradiction, that $W$ does not induce a clique.  This means that there are at most $\binom{r}{2} - 1$ edges of $G$ with both endpoints in $W$. In particular, there is woman of type 5 who is matched to a man of type 1, corresponding to some edge $e$ whose endpoints are not both in $W$.  Suppose, without loss of generality, that $v$ is an endpoint of $e$ that is not in $U$.  Since $v \notin W$, we know that $M(v)$ is of type 3. This means that $v$ strictly prefers $e$ (of type 1) to her current partner.  Moreover, since $v$ is an endpoint of $e$, she is considered exceptional by $e$ and so $e$ prefers $v$ to any woman of type 5, including $M(e)$. Thus, $(v,e)$ is a blocking pair, contradicting the assumption that $M$ is a stable matching.  We therefore conclude that $W$ must induce a clique and so $(G,r)$ is a yes-instance.
\end{proof}

The next result then follows immediately.

\begin{corollary}\label{2-exception-max-smti-hard}
\textsc{(2,Any)-Exception Typed Max SMTI} is \cplxty{NP}-hard, even if only men have exceptions in their preference lists, preferences over types are strict, and there are three types each of men and women.
\end{corollary}

\section{Hospitals$/$Residents problem with Couples}\label{sec:max-hrc}

The job market for medical residents underwent a change since mid 1970's, due to married couples seeking posts in nearby hospitals. Subsequently, the central matching systems had to be adapted to take into account \emph{couple's preferences}, as otherwise couples would seek to arrange their own matches outside the centralized clearinghouse. The \emph{Hospitals$/$Residents problem with Couples (HRC)} models such settings, where couples can send in preferences over pairs of hospitals. The size of a given matching in an instance of HRC is equal to the number of residents matched in the matching.

\begin{definition}\label{def:stable-hrc}
\textsc{Stable HRC} is the problem of deciding whether an instance of HRC admits a stable matching.
\end{definition}

\begin{definition}\label{def:max-hrc}
\textsc{Max HRC} is the problem of identifying a maximum cardinality stable matching in an HRC instance, or reporting that none exists.
\end{definition}

\textsc{Stable HRC} is NP-complete, and as a corollary so is \textsc{Max HRC}. The result holds even if each hospital has capacity 1 and there are no single residents \cite{Ron90}. The parameterized complexity of \textsc{Stable HRC} has been studied in \cite{MS11,BIS11}, where \cite{MS11} also studies \textsc{Max HRC}. \textsc{Stable HRC} is W[1]-hard when the problem is parametrized by the number of couples \cite{MS11}. The result holds even if each hospital has capacity 1. However, \textsc{Stable HRC} belongs to FPT when the problem is parametrized by the number of couples in the instance and the hospitals' list are derived from a strictly-ordered master list of residents \cite{BIS11}.

In the remainder of this section, we first provide a stability definition for instances of HRC, then extend our \textsc{Typed} and \textsc{Refined-typed} models to such instances, and conclude with presenting an FPT algorithm for \textsc{Typed Max HRC}.

\subsection{Stability in HRC}
In an instance of HRC, the set of residents $R$ includes an even size subset $R'$ consisting of those residents who belong to couples, where every resident in $R'$ belongs to exactly one couple. 
Let $C$ denote the set of ordered pairs $(r_i,r_j)$ where $r_i$ and $r_j$ form a couple. 
Single residents and hospitals find a subsets of candidates acceptable and rank them in strict order of preference. Every couple $(r_i,r_j) \in C$ has a joint strict preference ordering over an acceptable subset of ordered hospital pairs. 

Different stability definitions for an instance of HRC have been provided in the literature, most of them distinct from one another (see, e.g., \cite{BIS11,KPR13,MS11,MM10,Can04,DM97,KKN09,Rot84}). In this paper we adopt the definition of \cite{MM10} and provide it as it has been given in \cite{Man13}. We say that a hospital $h_j$ is undersubscribed if $|M(h_j)| < q(h_j)$ . Given an instance of HRC, a matching $M$ is stable if it admits no blocking pair, where a blocking pair satisfies one of the following properties:
\begin{enumerate}
	\item it involves a single resident $r_i$ and a hospital $h_j$ where (a) $r_i$ prefers $h_j$ to $M(r_i)$ and (b) $h_j$ is undersubscribed or prefers $r_i$ to its worst assigned resident in $M$;
    \item it involves a couple $(r_i,r_j) \in C$  and a hospital $h_k$ such that either 
    	\begin{enumerate}[(a)]
        	\item $(r_j,r_j)$ prefers $(h_k,M(r_j))$ to $(M(r_j),M(r_k))$, and $h_k$ is undersubscribed or prefers $r_i$ to a resident in $M(h_k) \setminus\{r_j\}$; or 
            \item $(r_j,r_j)$ prefers $(M(r_i), h_k)$ to $(M(r_j),M(r_k))$, and $h_k$ is undersubscribed or prefers $r_j$ to a resident in $M(h_k) \setminus\{r_i\}$;
        \end{enumerate}
    \item it involves a couple $(r_i,r_j) \in C$ and a pair of (not necessarily distinct) hospitals $h_k,h_{\ell}$ such that $h_k \neq M(r_i)$, $h_{\ell} \neq M(r_j)$, $(r_j,r_j)$ prefers $(h_k,h_{\ell})$ to $(M(r_j),M(r_k))$, and either
    \begin{enumerate}[(a)]
    	\item $h_k \neq h_{\ell}$, and $h_k$ (respectively $h_{\ell}$) is either undersubscribed or prefers $r_i$ (respectively $r_j$) to at least one of its assigned residents; or
        \item $h_k = h_{\ell}$ and $q(h_k) - |M(h_k)| \geq 2$; or
        \item $h_k = h_{\ell}$ and $q(h_k) - |M(h_k)| \geq 1$ and $h_k$ prefers at least one of $r_i, r_j$ to one of its assigned residents; or
        \item $h_k = h_{\ell}$, $q(h_k) = |M(h_k)|$, $h_k$ prefers $r_i$ to some resident $r_s \in M(h_k)$, and prefers $r_j$ to some resident in $M(h_k)$ distinct from $r_s$.
    \end{enumerate}
\end{enumerate}

\subsection{\textsc{Typed} Instances}
In an instance of HRC, each type $i$ that is associated with a single resident or a hospital has a preference ordering over types of the candidates, similarly to that defined earlier for instances of HRT. We assume, without loss of generality, that each pair of residents that form a couple are ordered such that the type of the first one is no larger than the type of the second one. That is, for each $(x,y) \in C$ where $x$ is of type $i$ and $y$ is of type $j$, it is the case that $i\leq j$.
Each pair of types $(i,j)$ (with $i\leq j$) associated with a couple  
has a joint preference ordering over ordered pairs of ``hospital'' types. 

The same conditions outlined in Section \ref{sec:typed} apply to single residents and hospitals. That is, (1) every two single residents or hospitals of the same type have identical preference lists, and (2) all single residents are indifferent between hospitals of the same type, and all hospitals are indifferent between residents of the same type.

These notions extend naturally to the preference lists of couples: 
\begin{enumerate}
\item if $(x,y)$ and $(x',y')$ are couples of the same type $(i,j)$ then $(x,y)$ and $(x',y')$ have identical preference lists, and
\item all couples are indifferent between $(h_s,h_t)$ and $(h_k,h_{\ell})$ if $h_s$ and $h_k$ are hospitals of the same type, say $i$, and $h_t$ and $h_{\ell}$ are hospitals of the same type, say $j$.
\end{enumerate}

\subsection{An FPT algorithm for \textsc{Typed MAX HRC}}
We adapt the method used to prove 
Theorem \ref{cor:typed-maxsrti-fpt} in order to show that \textsc{Typed Max HRC} is in $\FPT$ when parameterised by the number of different types in the instance. The idea, once again, is to (1) consider a number of possibilities for the matching, where this number of possibilities is bounded by a function of $k$, (2) determine for each such possibility whether a matching which meets the conditions will be stable, and then (3) express the maximisation problem associated with each set of stable conditions as an instance of \textsc{Integer Linear Programming}. 

As in previous sections, we assume without loss of generality that there is at least one agent of each type; however we cannot be sure that there will always be a \emph{single} resident of a given type, or a couple of a specific pair of types. Also, as in previous sections, we assume that every agent is matched by creating sufficiently many dummy agents of type $k+1$ which are inserted at the end of each single resident's and hospital's (possibly incomplete) preference list; we insert $(k+1,k+1)$ at the end of each couple's preference list.

We start by providing conditions that are necessary and sufficient for a matching to be stable in a given typed instance of HRC. Given a matching $M$ (in which every agent is matched, perhaps to a dummy agent), we define three functions $\worst_M$, $\secondworst_M$ and $\assigned_M$. The function $\worst_M$ is defined as follows:
\begin{itemize}
\item for any resident type $i$ of which there is at least one single agent, $\worst_M(i)$ is the type of the least desirable hospital (with respect to the preference list for $i$) to which any \emph{single} resident of type $i$ is assigned in $M$;
\item for any pair of resident types $i$ and $j$ (with $i\leq j$) such that at least one couple has type $(i,j)$,
$\worst_M(i,j)$ is the least desirable pair of hospital types (with respect to the joint preference list for $(i,j)$) to which any couple of types $i$ and $j$ are assigned in $M$;
\item for any hospital type $p$, $\worst_M(p)$ is the type of the least desirable resident (with respect to the preference list for $p$) assigned to a hospital of type $p$ in $M$.
\end{itemize}
The function $\secondworst_M$ is only defined for hospital types where the total capacity of hospitals of the type is at least two: $\secondworst_M(p)$ is the type of the second least desirable resident assigned to any hospital of type $p$ in $M$.  The boolean function $\assigned_M$ is defined for each combination of a pair of resident types and a pair of hospital types: $\assigned_M((i,j),(\ell,p)) = 1$ if and only if there is at least one couple involving residents of type $i$ and $j$ who are assigned to hospitals of types $\ell$ and $p$ respectively. 

For the remainder of this section, we assume that types $1$ to $k'$ are types of residents, and types $k'+1$ to $k$ are types of hospitals. Let $(k':k]$ denote the set $\{k'+1,\ldots k\}$. Let $|N_i|$ (with $i\in[k']$) denote the number of single residents of type $i$, $|N_p|$ (with $p \in (k':k]$) denote the total capacity of hospitals of type $p$, and $|N_{i,j}|$ (with $i,j \in [k']$ where $i\leq j$) the number of couples where the first resident is of type $i$ and the second one is of type $j$.

It is straightforward to prove the following lemma using a similar approach as in the proof of \ref{stability-test-srti} and the definition of a stable matching for an instance of HRC.

\begin{lemma}\label{hrc-stability-test}
Let $I$ be a typed instance of HRC. Then a matching $M$ is stable if and only if none of the following holds:
\begin{enumerate}
	\item There is a pair of types $(i,p)$, $i \in [k'], p \in (k':k]$ such that $|N_i|>0$, $p \pref_i \worst_M(i)$ and $i \pref_p \worst_M(p)$ .
	\item There is a couple-type $(i,j)$ with $|N_{i,j}| > 0$ and a hospital type $\ell$ such that either
	\begin{enumerate}
		\item $i \succ_{\ell} \worst_M(\ell)$ and there is a pair of hospital types $(p,q)$ such that $\assigned_M((i,j),(p,q)) = 1$ and $(\ell, q) \pref_{i,j} (p,q)$, or 
		\item $j \succ_{\ell} \worst_M(\ell)$ and there is a pair of hospital types $(p,q)$ such that $\assigned_M((i,j),(p,q)) = 1$ and $(p,\ell) \pref_{i,j} (p, q)$.
	\end{enumerate}
	\item There is a couple-type $(i,j)$ with $|N_{i,j}| > 0$ and two hospital types $\ell$ and $p$ (where $\ell \neq p$) such that (i) $(\ell,p) \pref_{i,j} \worst_M(i,j)$, (ii) $i \succ_{\ell} \worst_M(\ell)$, and (iii) $j \succ_{p} \worst_M(p)$;
\item There is a couple-type $(i,j)$ with $|N_{i,j}| > 0$ and a hospital type $\ell$ such that (i) $(\ell,\ell) \pref_{i,j}  \worst_M(i,j)$, (ii) $i,j \succ_{\ell} \worst_M(\ell)$, and either
	\begin{enumerate}
		\item $i \succ_{\ell} \secondworst(\ell)$, or
		\item $j \succ_{\ell} \secondworst(\ell)$.
	\end{enumerate}
\end{enumerate}
\end{lemma}

\begin{theorem}
\textsc{Typed Max HRC} is in $\FPT$ when parameterised by the number $k$ of different types in the instance.
\end{theorem}
\begin{proof}
Lemma ~\ref{hrc-stability-test} gives a necessary and sufficient condition for a matching realising the functions $\worst$, $\secondworst$ and $\assigned$ to be stable. We consider each of the feasible possibilities for the functions $\worst$, $\secondworst$ and $\assigned$ in turn and determine, using Lemma ~\ref{hrc-stability-test}, whether a matching realising them will be stable; this can be done for any set of candidate functions in time $\mathcal{O}(k^3)$.

For any trio of feasible functions $\worst$, $\secondworst$ and $\assigned$ that give rise to a stable matching (if any), we need to determine the maximum number of residents that can be matched in any matching that realises the trio, and we do this (as in the rest of the paper) by solving an instance of \textsc{Integer Linear Programming}. 

We define variables $n_{i,p}$ and $n_{(i,j),(p,q)}$ as follows. 
For each pair of types $(i,p)$, $i \in [k'], p \in (k':k]$, we have a variable $n_{i,p}$ which denotes the number of \emph{single} residents of type $i$ matched to hospitals of type $p$.  Additionally, for each combination of a pair of resident types $(i,j)$ (with $i \leq j$) and a pair of hospital types $(p,q)$ (where $p$ and $q$ may or may not be distinct), we have a variable $n_{(i,j),(p,q)}$ which denotes the number of couples consisting of residents of types $i$ and $j$ respectively such that the resident of type $i$ is assigned to a hospital of type $p$ and the resident of type $j$ is assigned to a hospital of type $q$. Our objective is to maximise the size of the matching hence to maximise:
\begin{align*}
~ ~ \sum_{i \leq k', k' < p \leq k} n_{i,p} + \sum_{i,j \leq k', k'< p,q \leq k} 2\cdot n_{(i,j),(p,q)} 
\end{align*}
Our first set of constraints ensures that every agent is involved in exactly one pair. The first constraint handles single agents, the second one the hospitals, and the last one couples.
\begin{align*}
~ &~ ~ \sum_{k'< p \leq k} n_{i,p} = |N_i|, & ~ \forall i \in [k'] \\
~ &~ ~ \sum_{i \leq k'} n_{i,p} + \sum_{i,j \leq k', k'< q \leq k}(n_{(i,j),(p,q)} + n_{(i,j),(q,p)})= |N_p|, & ~ \forall p \in (k':k] \\
~ &~ ~ \sum_{k'< p,q \leq k} n_{(i,j),(p,q)} = |N_{i,j}|, & ~ \forall i, j: i\leq j \leq k'\\
\end{align*}

The next set of constraints ensures that function $\worst$ complies with its definition.
\begin{align*}
~ & \sum_{p\weaklypref_i \worst(i)} n_{i,p} = |N_i|, & ~ \forall i \in [k']: |N_i| >0 \\[6pt]
~  & n_{i,\worst(i)} > 0, & ~ \forall i \in [k']: |N_i| >0 \\[12pt]
~ & \sum_{i \weaklypref_p \worst(p)} n_{i,p} + \sum_{\substack{i,j \weaklypref_p \worst(p) \\ k'< q \leq k}} (n_{(i,j),(p,q)} + n_{(i,j),(q,p)}) = |N_p|, & ~ \forall p \in (k':k] \\[6pt]
~  & n_{\worst(p),p} > 0, & ~ \forall p \in (k':k] \\[12pt]
~ & \sum_{(p,q) \weaklypref_{(i,j)} \worst(i,j)} n_{(i,j),(p,q)} = |N_{i,j}|, & ~ \forall i, j: i\leq j \leq k': |N_{i,j}| >0\\[6pt]
~  & n_{(i,j),\worst(i,j)} > 0, & ~ \forall i,j: i\leq j \leq k': |N_{i,j}| >0 \\
\end{align*}

The following set of constraints ensures that function $\secondworst$ for the hospitals complies with its definition.
\begin{align*}
~ & \sum_{\secondworst(p) \pref_p i \pref_p \worst(p)} n_{i,p} = 0, & ~ \forall p \in (k':k]: |N_p| >1 \\
~  & n_{\worst(p),p} + n_{\secondworst(p),p} > 1, & ~ \forall p \in (k':k]: |N_p| >1 \\
~ & n_{\worst(p),p} > 1 & ~ \forall p \in (k':k]: |N_p| >1 \text{ and } \secondworst(p) = \worst(p)
\end{align*}

And finally, the last set of constraints ensure that boolean variables $\assigned$ are set correctly.
\begin{align*}
~ & n_{(i,j),(p,q)} > 0, & ~ \forall i,j \in [k'] \forall p,q \in (k',k] \text{ such that } \assigned((i,j)(p,q)) = 1\\
\end{align*}

The above integer linear program has $\mathcal{O}(k^4)$ variables and $\mathcal{O}(k^2)$ constraints. The upper bound on the absolute value a variable can take is $n$. Therefore, by Theorem \ref{thm:ILP}, this maximisation problem for any trio of candidate functions $\worst$, $\secondworst$ and $\assigned$ can be solved in time $2^{\mathcal{O}(k^2)}\log^3 n$.
\end{proof}

\section{Summary and Future Work}\label{sec:future-work}

We studied settings in which agents are partitioned into $k$ different types, 
and the type of an agent determines his or her preferences, as well as how s/he is compared against other agents. We considered a basic setting and two generalisations.
In the basic model, referred to as typed, agents have preferences over types and agents of the same types have identical preferences. In the first generalisation, referred to as consistently-refined-typed, agents' preferences may be refined by more detailed preferences within a single type, subject to the requirement that agents of the same type still have identical preference lists.  
In the second generalisation, referred to as exception-typed, each agent may regard some small collection of other agents to be exceptional and rank them without regard to their types.

If $k$ is considered to be part of the input, then our models do not place any restrictions on preference lists; we can create a type for each agent and hence let $k=n$. 
Thus we can deduce that \textsc{Max SMTI}, \textsc{Max HRT}, \textsc{Max SRT}, \textsc{Max Size Min BP/BA SMI}, \textsc{Max Size Min BP/BA SRI}, and \textsc{Min BP SR}
are all \cplxty{NP}-complete, in either of our three settings, when $k$ is part of the input.

We are interested in, and have argued for the existence of, scenarios where $k$ is not part of the input and is much smaller than $n$.
Under both typed and consistently-refined-typed settings, we showed that
\textsc{Max SMTI}, \textsc{Max HRT}, \textsc{Max SRTI}, \textsc{Max Size Min BP/BA SMTI}, \textsc{Max Size Min BP/BA SRTI}, and \textsc{Min BP SRTI}
belong to the parameterised complexity class FPT when parameterised by the number of different types of agents, and so admit efficient algorithms when this number of types is small. We were further able to prove that \textsc{Max SMTI}, \textsc{Max HRT} and \textsc{Max SRTI} are polynomial-time solvable when agents have strict preferences over types. 
Additionally, we were able to show that, under the typed setting, \textsc{Max HRC} is in FPT parameterised by the number of different types of agents.

Under the exception-typed setting, we showed that if each agent finds one candidate exceptional and promotes that candidate to the top of his/her preference list, then \textsc{Max SMTI} belongs to FPT parameterised by $k$. In contrast, if we allow for each agent to find two or more candidates exceptional who can appear anywhere in the preference lists, then \textsc{Max SMTI} remains \cplxty{NP}-hard, even when the number of types is bounded by a constant.

It would be interesting to investigate what further generalisations of our model yield \cplxty{FPT} algorithms for \cplxty{NP}-hard stable matching problems.  In particular, the complexity of \textsc{(1,Bottom)-Exception-Typed Max SMTI}, \textsc{(1,Any)-Exception-Typed Max SMTI}, and \textsc{(2,Top)-Exception-Typed Max SMTI} remain open, as well as the complexity of \textsc{Max SRTI}, \textsc{Max Size Min BP/BA SMTI}, \textsc{Max Size Min BP/BA SRTI}, and \textsc{Min BP SRTI} under the exception-typed setting. Moreover, we could consider further restrictions with two or more exceptions, for example if an exceptional candidate can only be moved to the top or bottom of its type.  

Another intriguing question would be to understand how the complexity of \textsc{MAX SMTI} and \textsc{Max Size Min BP/BA SMTI} changes when agents on only one side of the market are associated with types. 

\paragraph{Acknowledgements.}

The first author is supported by a Personal Research Fellowship from the Royal Society of Edinburgh (funded by the Scottish Government).  Both authors are extremely grateful to David Manlove for his insightful comments on a preliminary version of this manuscript.

\bibliographystyle{plain}
\bibliography{bibfile}

\end{document}